\documentclass[12pt]{article}
\usepackage[margin=1in]{geometry} % see geometry.pdf on how to lay out the page. There's lots.
\usepackage{setspace}
\usepackage{indentfirst}
\usepackage{chngcntr}
\usepackage{multirow}
\usepackage{amsmath,amsthm,amssymb,url}%setspace,fancyhdr,geometry}
\usepackage{graphicx}%psfig,subfigure
\newtheorem{thm}{Theorem}

\usepackage[affil-it]{authblk}

\newtheorem{lemma}{Lemma}
\newtheorem{corollary}{Corollary}

\DeclareMathOperator*{\argmin}{arg\,min}
\usepackage{tikz}
\usetikzlibrary{shapes.misc}
\usetikzlibrary{shadows}
\usetikzlibrary{arrows.meta,positioning}
\newdimen\nodeDist
\newdimen\nodeDistdata
\makeatletter

\@addtoreset{thm}{appendix}
\makeatother

\makeatletter
\renewcommand\@biblabel[1]{}
\makeatother
\usepackage{color}
\newcommand\correspondingauthor{\thanks{Corresponding author. Email address : ycho@utep.edu}}
\title{Analysis of regression discontinuity designs using censored data}
\author[1]{Youngjoo Cho\correspondingauthor}
\date{}
\author[2]{Chen Hu}
\author[3]{Debashis Ghosh}
\affil[1]{Department of Mathematical Sciences, The University of Texas at El Paso, El Paso, TX 79912}
\affil[2]{Division of Biostatistics and Bioinformatics, Sidney Kimmel Comprehensive Cancer Center, Johns Hopkins University School of Medicine, Baltimore, MD 21205}
\affil[3]{Department of Biostatistics and Informatics, University of Colorado Anschutz Medical Campus, Aurora, CO 80045}
\begin{document}
\maketitle
\begin{abstract}
In medical settings, treatment assignment may be determined by a clinically important covariate that predicts patients' risk of event. There is a class of methods from the social science literature known as regression discontinuity (RD) designs that can be used to estimate the treatment effect in this situation.  Under certain assumptions, such an estimand enjoys a causal interpretation. However, few authors have discussed the use of RD for censored data. In this paper, we show how to estimate causal effects under the regression discontinuity design for censored data. The proposed estimation procedure employs a class of censoring unbiased transformations that includes inverse probability censored weighting and doubly robust transformation schemes. Simulation studies demonstrate the utility of the proposed methodology.
\end{abstract}
\section{Introduction} 
\indent In observational studies, scientific interest typically focuses on estimation of a causal estimand.  The presence of confounding variables makes its estimation difficult.  To perform causal inference, the analyst typically relies on several assumptions. One important assumption is the `no unmeasured confounders' assumption, which implies that treatment assignment is independent of potential outcome given confounders. This has also been referred to as the unconfoundedness assumption. However, this assumption is typically not empirically testable. \\
\indent A study design that has been used in the social sciences is known as the regression discontinuity (RD) design.    
One appealing feature to the RD design is that the treatment assignment is determined either deterministically or probabilistically by continuous variable of interest, termed the forcing variable.   It turns out for such a design, the no unmeasured confounders assumption is not required for inferring causality. The study of RD designs was initiated by Thistlethwaite and Campbell (1960) and has been developed further in many subsequent studies. For example, Hahn, Todd and Van der Klaauw (1999) and Hahn, Todd and Van der Klaauw (2001) prove theoretical results on RD designs. Ludwig and Miller (2005) and Ludwig and Miller (2007) propose bandwidth selection procedures for nonparametric estimation with application to evaluating the effects of funding on educational program. Imbens and Kalyanaraman (2012). propose an optimal bandwidth selection procedure. Recently, Calonico, Cattaneo and Titiunik (2014) have developed bias-corrected nonparametric estimation approaches whose confidence intervals demonstrate improved coverage relative to those from other RD estimators. \\
\indent Much of the work in the previous paragraph dealt with the case of uncensored data.  In many biomedical settings, the outcome will be a time to event that is potentially subject to right censoring.  For this problem, there has been limited work in the area of RD designs.  Recently, Bor et al. (2014) and Moscoe, Bor and B\"{a}rnighausen (2015) discuss the use RD designs in medical and epidemiological studies. In their paper, the goal is to test for the effect of early versus late treatment initiation to HIV patients on survival. 
\\
\indent In practice, RD designs are useful for detection of treatment effect in the studies with censored data. For example, in the Prostate, Lung, Colorectal and Ovarian (PLCO) cancer study, the role of prostate cancer screening remains controversial partially due to the findings of PLCO which failed to show reduction in prostate cancer-specific survival or overall mortality by a prostate-specific antigen (PSA) screening strategy (Andriole et al. 2009). Shoag et al. (2015) consider only the treatment group and propose that a PSA level $\geq$ 4.0mg/nl is a reasonable cutoff for some outcomes by using RD design, but they did not consider the censored nature of the outcome variable in their analysis. It would be of great interest to clinicians for understanding the screening effect of PSA level 4.0mg/nl on time to death or to first cancer incidence. \\
\indent  For the use of RD designs in survival analysis, one challenge is the existence of censoring. Applying standard nonparametric estimation procedures from the uncensored data setting is not feasible. %There are two possible solutions for this issue.   The first is to fit a survival model, such as the well-known Cox proportional hazards model, to estimate the treatment effect.  
One way to solve this issue is to use existing RD estimation procedures with some transformation of the response that behaves in a manner analogous to the uncensored data case. %This approach is advantageous in that it is less parametric than a model-based approach. 
Fan and Gijbels (1994) propose local linear regression based on a transformed response for censored data. Their proposed transformation includes the inverse probability weighted censoring (IPCW) method, a commonly used method in the missing data literature to handle censoring. However, this method is inefficient in that it does not include information of censored observations in estimation. Rubin and Van der Laan (2007) overcome this difficulty by proposing a doubly robust transformation of the response which requires modeling of failure time distribution as well as censoring distribution. This approach shows promise compared to IPCW methods in a prognostic modeling  (Steingrimsson et al. 2016; Steingrimsson, Diao and Strawderman, 2019), but no studies have shown its efficiency gains for estimation of parameters in regression modeling with purpose of inference. \\
\indent In this paper, we propose a class of estimation procedures in the RD design for censored data. In Section 2, we review the relevant data structures and discuss approaches of RD design for uncensored data. Section 3 and 4 describe the extension of RD designs to censored data as well as laying out the methodology with attendant asymptotic results.  In Section 5, simulation studies are performed which show the utility of our approach.  In Section 6, we apply our method to the PLCO dataset to test effect of treatment assignment by PSA.  Some discussion concludes Section 7. 
%Some discussion concludes Section 6.       
%\section{RD design}   
\section{Review of RD design for uncensored data} 
%%%% START HERE!!!!
Before discussing the proposed methodology, we first introduce the RD design for uncensored data using a potential outcomes framework. Let $(Y_*^{(1)},Y_*^{(0)})$ be the potential outcomes under treatment and control; we use $Z$ to define treatment. We define ${\bf W}$ to be a vector of forcing variables; for simplicity, we consider only one forcing variable $W$.  As the name suggests, the forcing variable determines the treatment (Imbens and Lemieux, 2008).   Since only one of the potential outcomes is observable, the observed response is $Y_* = ZY_*^{(1)} + (1-Z)Y_*^{(0)}$. The main characteristic of the RD design is that the treatment assignment $Z$ depends on a function of $W$, which can be deterministic or probabilistic.  This corresponds to the sharp and fuzzy RD designs, respectively. \\
\indent One key assumption in RD designs is that the there is no gaming of the forcing variable $W$ (McCrary, 2008).  This is checked in practice by empirically plotting the distribution of $W$ and checking to see that there is no `clumping' around the cutoff value of interest.     
%One important assumption in RD design is that subjects do not have perfect control on the forcing variable. Hence, 
When the forcing variable is not affected by other variables, we have a `locally randomized' study based on the cutpoint of the forcing variable from this assumption (Lee and Lemieux, 2010). This `local randomization' is a compelling feature compared to usual observational studies and allows for establishing causality as in randomized experiments (Bor et al. 2014). \\ 
\indent In the sharp RD design, treatment assignment is decided by a deterministic function of forcing variable. Let $H^*$ be a known discontinuous function. Then for the sharp RD, $Z = H^*(W)$. The main causal effect of interest is the average treatment effect at the discontinuity point $w_0$. By design, if the value of the forcing variable is greater than or equal to the cutpoint, $E\{Y_*^{(1)}|W\} = E(Y_*|W)$.  Similarly, $E\{Y_*^{(0)}|W\} = E(Y_*|W)$ if the value of forcing variable is less than cutpoint. Since our interest focuses on the causal effect at $w_0$, with a continuity assumption for $E\{Y_*^{(1)}|W\}$ and $E\{Y_*^{(0)}|W\}$, we can identify limits around the threshold (Bor et al. 2014). 
\begin{gather}
\label{eq:1a}
E[Y_*^{(1)} - Y_*^{(0)}|W = w_0] = \lim_{w \downarrow w_0} E(Y_*|W = w) - \lim_{w \uparrow w_0} E(Y_*|W = w).
\end{gather}
\indent In the fuzzy RD design, treatment assignment is a probabilistic function of the forcing variable. There exists a jump in the probability of treatment assignment at the threshold (Imbens and Lemieux, 2008) but it is less than one and depends on the forcing variable. Let $q_1$ and $q_2$ be monotone functions of the forcing variable. For fuzzy RD designs, the probability of receiving treatment is 
\begin{gather*}
P(Z = 1|W) = 
			\begin{cases}
			q_1(W) \ \text{if} \ W < w_0 \\
			q_2(W) \ \text{if} \ W \geq w_0.
			\end{cases}
\end{gather*}
By contrast, for sharp RD designs, the probability of receiving treatment given that forcing variable is greater than certain cutoff is one. Hence the treatment status and treatment assignment are the same. In fuzzy RD designs,  
\begin{gather*}
\lim_{w \downarrow w_0} P(Z=1|W=w) \ne \lim_{w \uparrow w_0} P(Z=1|W=w)
\end{gather*}
with the difference of $\lim_{w \downarrow w_0} P(Z=1|W=w)$ and $\lim_{w \uparrow w_0} P(Z=1|W=w)$ not being equal to one. Hence, for fuzzy RD, the treatment assignment is not equivalent to the treatment status. One can consider the sharp RD as a special case of the fuzzy RD when $\lim_{w \downarrow w_0} P(Z=1|W=w) - \lim_{w \uparrow w_0} P(Z=1|W=w)=1$. Now expression (2.1) is not the average treatment effect in a fuzzy RD design. However, since the forcing variable determines treatment assignment, it effectively functions as an instrumental variable. By using the arguments in Angrist, Imbens and Rubin (1996), we can obtain the so-called complier average treatment effect.
%effect for people whose treatment assignment coincides with the treatment status. Such people are called compliers. 
This effect is the main identifiable causal estimand in the fuzzy RD design and is important in practice because participants in the study may not comply with the initial treatment assignment. Formally, we have 
\begin{gather*}
\frac{\lim_{w \downarrow w_0} E(Y_*|W = w) - \lim_{w \uparrow w_0} E(Y_*|W = w)}{\lim_{w \downarrow w_0} P(Z|W = w) - \lim_{w \uparrow w_0} P(Z|W = w)} = E[Y_*^{(1)} - Y_*^{(0)}|W = w_0,\text{subject is a complier}]
\end{gather*}   
In this case, the unconfoundedness assumption is not realistic because people with similar values of the forcing variable away from the cutpoint might receive different treatment. These two subjects will thus not be comparable (Imbens and Lemieux, 2008).

\section{Extension of RD designs to censored data: data structure and assumptions}
As in the uncensored data case, we define $(T^{(1)},T^{(0)})$ as potential outcomes under treatment and control assignments, respectively. In survival data, $(T^{(1)},T^{(0)})$ are potential times to event for the treatment and control. Without censoring, time to failure $T$ is only observable for either of treatment and control group, i.e. $T = ZT^{(1)} + (1-Z)T^{(0)}$. Let $X = (T,W,Z)$ as full data. Let $C$ be time to censoring. We can only observe $\tilde{T} = T \wedge C$, $\Delta = I(T \leq C)$.  Define $O = (\tilde{T},\Delta,Z,W)$. The observable data are $O_i = (\tilde{T}_i,W_i,Z_i,\Delta_i),i=1,\ldots,n$. We assume that the observed data are independent and identically distributed. We apply a logarithm transform to the response.  Define 
\begin{gather*}
Y^{(1)} = \log T^{(1)} \quad Y^{(0)} = \log T^{(0)} \quad Y = \log T \quad \tilde{Y} = Y \wedge \log C
\end{gather*}
and $Y_i$ and $\tilde{Y}_i$ are individual realizations of $Y$ and $\tilde{Y}$, respectively. Usually, the censoring time is not affected by treatment assignment and status so that it is reasonable to assume that $C^{(1)}=C^{(0)}=C$ (Bai, Tsiatis, and O'Brien, 2013), where $(C^{(1)},C^{(0)})$ are the potential outcomes for censoring under treatment and control. Let $w_0$ be a cutpoint for the forcing variable. We now discuss necessary assumptions for identification of causal effects in RD designs with censored data. 
\begin{enumerate}
\item $E(Y^{(1)} | W=w)$ and $E(Y^{(0)} | W=w)$ are continuous in $w$. 
\item $Y^{(0)}, Y^{(1)},W, Z$ are independent with $C$.
%\item Censoring is random, $C \perp W$
\item Let $S^{(1)}(t|w) = P(T^{(1)} > t|W=w)$ and $S^{(0)}(t|w) = P(T^{(0)} > t|W=w)$. $S^{(1)}(t|w)$ and $S^{(0)}(t|w)$ are continuous at $w$ for all $t$. Denote $S(t|w) = P(T > t|W=w)$.
\item Let $G(t) = P(C>t)$. $G(t)$ is continuous for all $t$.
%\item By the discussion in previous paragraph, $C$ is independent of $Z$.
%\item For sharp RD, $T^{(0)}, T^{(1)}$ are conditionally independent with $Z$ given $W$.
\item For the fuzzy RD, $P(Z=1|W=w)$ is continuous at $w$ except for $w = w_0$ and $\lim_{w \downarrow w_0} P(Z=1|W=w) \ne \lim_{w \uparrow w_0} P(Z=1|W=w)$
\item For the fuzzy RD, $Z(w^*)$ is nondecreasing in $w^*$ at $w^*=w_0$ where $Z(w^*)$ is potential treatment status given point $w^*$ such that $w^*$ is neighborhood of $w_0$. 
\end{enumerate}
Condition 1 is a smoothness assumption for the mean potential outcome functions around the cutoff point $w_0$. Condition 2 is a typical assumption in survival analysis and is referred to as random censoring. Conditions 3 and 4 guarantee smoothness for the failure time and censoring distributions. %(5) is a usual assumption in causal inference literature. It states that there are no unmeasured confounders, However, for fuzzy RD, since entire members in treatment and control group are not comparable in $w_0$, this assumption is not valid. 
Condition 5 is a standard assumption for fuzzy RD. Condition 6 states that potential treatment status is monotonic in the cutoff point (Imbens and Lemieux, 2008). For survival data, in the sharp RD design, the average causal effect of treatment given the forcing variable is $E(Y^{(1)} -Y^{(0)} | W=w)$. Under these assumptions,
\begin{gather*}
E\{Y^{(0)}|W=w_0\} = \lim_{w \uparrow w_0} E\{Y^{(0)}|Z = 0,W=w_0\} = \lim_{w \uparrow w_0} E(Y|W=w) \\
E\{Y^{(1)}|W=w_0\} = \lim_{w \downarrow w_0} E\{Y^{(1)}|Z = 1,W=w_0\} = \lim_{w \downarrow w_0} E(Y|W=w)
\end{gather*}
Then under a sharp RD, the average treatment effect is 
\begin{gather*}
\tau_{SRD} = \lim_{w \downarrow w_0} E[Y|W=w] - \lim_{w \uparrow w_0} E[Y|W=w]  
\end{gather*}
For a fuzzy RD, the average treatment effect is
\begin{gather*}
\tau_{FRD} = \frac{\lim_{w \downarrow w_0} E(Y|W=w) - \lim_{w \uparrow w_0} E(Y|W=w)}{\lim_{w \downarrow w_0} P(Z=1|W=w) - \lim_{w \uparrow w_0} P(Z=1|W=w)}
\end{gather*}

\section{Proposed Methodology}
\subsection{Censoring unbiased transformations} 
In sharp RD with uncensored data, we can use nonparametric estimation directly with the response variable. However, for censored data, it is very difficult to use response directly due to censoring. The issue is to find a transformation $q^*$ based on observed data such that $E\{q^*(O)\} = E(H(T)|W)$ where $H(\cdot)$ is a nondecreasing function of $T$. The transformation is called a censoring unbiased transformation (Fan and Gijbels, 1994; Rubin and Van der Laan, 2007; Steingrimsson, Diao, and Strawderman, 2019). One approach is to use an inverse probability censoring weighted (IPCW) method to obtain $E(H(T)|W)$, which is 
\begin{gather*}
H_{IPCW}(T) = \frac{\Delta H(T)}{G(T|W)}.
\end{gather*}
%Due to its easy implementation, this approach is widely used in survival analysis. In particular, only the censoring distribution needs to be modelled.  However, this approach requires that the censoring distribution is correctly specified.  In addition, censored observations are not used in computing the transformation. Hence this approach yields an inefficient estimator. Rubin and Van der Laan (2007) propose a doubly robust (henceforth DR) censoring unbiased transformation for the failure time. Steingrimsson, Diao, and Strawderman (2019) extend this approach to $H(T)$. Let 
It is easy to show that $H_{IPCW}(T)$ is censoring unbiased transformation. However, this approach requires that the censoring distribution is correctly specified and this approach yields an inefficient estimator. Rubin and Van der Laan (2007) propose a doubly robust (henceforth DR) censoring unbiased transformation for the failure time. Steingrimsson, Diao, and Strawderman (2019) extend this approach to $H(T)$. Let 
\begin{gather*}
M_G(u|W) = I(\tilde{T} \leq u, \Delta=0) - \int_0^{u} I(\tilde{T} \geq s)\lambda_G(s|W)ds,
\end{gather*}
where $\lambda_G(s|W)$ is true hazard function of $G$ given covariate $W$. The form is
\begin{gather*}
H_{DR}(T) = \frac{\Delta H(T)}{G(T|W)} + \int_0^{\tilde{T}} \frac{Q_{H(T)}(u,W)}{G(u|W)}dM_G(u|W).
\end{gather*}
where $Q_{H(T)}(u,W)$ is $E\{H(T)|T > u,W\}$. This DR transformation requires estimation of both censoring and failure time distributions. It is a combination of an IPCW term and a mean zero martingale transform term. IPCW term is a special case of DR transformation. This martingale transformation term utilizes information from censored observations, which yields greater efficiency than using the IPCW approach only. Since it is also a censoring unbiased transformation, it guarantees that $E(H_{DR}(T)|W=w) = E(H(T)|W=w)$ for any $w$ (Steingrimsson, Diao, and Strawderman, 2019). Moreover, if either the censoring distribution or failure time distribution is correctly specified, then the estimator from DR transformation based on the estimated $Q$ and $G$ is statistically consistent for $E(H(T)|W)$. Moreover, if models for $Q$ and $G$ are both correctly specified, then the resulting estimator from the DR transformation is the most efficient estimator given class of the estimator of $E(H(T)|W)$ (Steingrimsson, Diao, and Strawderman, 2019). 
\subsection{Asymptotic Theory}
The key idea of using the transformation from Section 4.1. is to be able to use uncensored data techniques in the censored situation by applying the data-dependent transformation to the observed outcome. As can be seen in the last section, these transformations depend on the censoring distribution or both the censoring and failure time distributions. In this section, we discuss RD-based estimation procedures of the treatment effect and their asymptotic properties with censored data. Here, we take $H(\cdot) =\log (\cdot)$. and
\begin{gather*}
Y_{IPCW,i}(O_i;G) = \frac{\Delta_i Y_i}{G(T_i)} \quad Y_{DR,i}(O_i;G,S) = \frac{\Delta_i Y_i}{G(T_i)} + \int_0^{\tilde{T}_i} \frac{Q_Y(u,W_i)}{G(u)}dM_{G,i}(u),
\end{gather*}
where $Q_Y(\cdot,\cdot) = E(\log T|T \geq \cdot,\cdot)$ and
\begin{gather*}
M_{G,i}(u) = I(\tilde{T}_i \leq u, \Delta_i=0) - \int_0^{u} I(\tilde{T}_i \geq s)\lambda_G(s)ds.
\end{gather*}
We only focus on DR transformation. To obtain limits in the sharp and fuzzy RD designs, parametric methods are not very attractive because modeling and the discontinuity point(s) depend on the particular parametric distribution. Hahn, Todd and Van der Klaauw (2001) show that a causal effect in RD designs is nonparametrically estimable.  Now we apply the local linear regression method of Fan and Gijbels (1996). Let $K(\cdot)$ be a kernel function and $h$ be bandwidth. For fuzzy RD, we consider 
\begin{gather*}
%U_{R,IPCW}^{FRD,Y}(\alpha_R^{Y},\beta_R^{Y};G) = \sum_{i=1}^n I(W_i \geq w_0)\{Y_{IPCW,i}(O_i;G) - \alpha_R^{(Y)} - \beta_R^{(Y)} (W_i-w_0)\}^2 K\bigg{(}\frac{W_i-w_0}{h}\bigg{)} \\
%U_{L,IPCW}^{FRD,Y}(\alpha_L^{Y},\beta_L^{Y};G) = \sum_{i=1}^n I(W_i < w_0)\{Y_{IPCW,i}(O_i;G) - \alpha_L^{(Y)} - \beta_L^{(Y)} (W_i-w_0)\}^2 K\bigg{(}\frac{W_i-w_0}{h}\bigg{)} \\
U_{R,DR}^{FRD,Y}(\alpha_R^{Y},\beta_R^{Y};G,S) = \sum_{i=1}^n I(W_i \geq w_0)\{Y_{DR,i}(O_i;G,S) - \alpha_R^{(Y)} - \beta_R^{(Y)} (W_i-w_0)\}^2 K\bigg{(}\frac{W_i-w_0}{h}\bigg{)} \\
U_{L,DR}^{FRD,Y}(\alpha_L^{Y},\beta_L^{Y};G,S) = \sum_{i=1}^n I(W_i < w_0)\{Y_{DR,i}(O_i;G,S) - \alpha_L^{(Y)} - \beta_L^{(Y)} (W_i-w_0)\}^2 K\bigg{(}\frac{W_i-w_0}{h}\bigg{)} \\
U_R^{FRD,Z}(\alpha_R^Z,\beta_R^Z) = \sum_{i=1}^n I(W_i \geq w_0)\{Z_i - \alpha_R^Z - \beta_R^Z (W_i-w_0)\}^2 K\bigg{(}\frac{W_i-w_0}{h}\bigg{)} \\
U_L^{FRD,Z}(\alpha_L^Z,\beta_L^Z) = \sum_{i=1}^n I(W_i < w_0)\{Z_i - \alpha_L^Z - \beta_L^Z (W_i-w_0)\}^2 K\bigg{(}\frac{W_i-w_0}{h}\bigg{)}. 
\end{gather*}
{We can similarly define loss functions for IPCW transformation, say $U_{R,IPCW}^{FRD,Y}(\alpha_R^{Y},\beta_R^{Y};G)$ and $U_{L,IPCW}^{FRD,Y}(\alpha_L^{Y},\beta_L^{Y};G)$. In sharp RD designs, since the treatment assignment is deterministic, the loss functions with transformation responses are only necessary and they are identical to ones in fuzzy RD design.} \\
%\begin{gather*}
%U_{R,IPCW}^{SRD,Y}(\alpha_R^{Y},\beta_R^{Y};G) = \sum_{i=1}^n I(W_i \geq w_0)\{Y_{IPCW,i}(O_i;G) - \alpha_R^{(Y)} - \beta_R^{(Y)} (W_i-w_0)\}^2 K\bigg{(}\frac{W_i-w_0}{h}\bigg{)} \\
%U_{L,IPCW}^{SRD,Y}(\alpha_L^{Y},\beta_L^{Y};G) = \sum_{i=1}^n I(W_i < w_0)\{Y_{IPCW,i}(O_i;G) - \alpha_L^{(Y)} - \beta_L^{(Y)} (W_i-w_0)\}^2 K\bigg{(}\frac{W_i-w_0}{h}\bigg{)} \\
%U_{R,DR}^{SRD,Y}(\alpha_R^{Y},\beta_R^{Y};G,S) = \sum_{i=1}^n I(W_i \geq w_0)\{Y_{DR,i}(O_i;G,S) - \alpha_R^{(Y)} - \beta_R^{(Y)} (W_i-w_0)\}^2 K\bigg{(}\frac{W_i-w_0}{h}\bigg{)} \\
%U_{L,DR}^{SRD,Y}(\alpha_L^{Y},\beta_L^{Y};G,S) = \sum_{i=1}^n I(W_i < w_0)\{Y_{DR,i}(O_i;G,S) - \alpha_L^{(Y)} - \beta_L^{(Y)} (W_i-w_0)\}^2 K\bigg{(}\frac{W_i-w_0}{h}\bigg{)}.
%{U_{R,trans}^{SRD,Y}(\alpha_R^{Y},\beta_R^{Y};G,S) = \sum_{i=1}^n I(W_i \geq w_0)\{Y_{trans,i}(O_i;G,S) - \alpha_R^{(Y)} - \beta_R^{(Y)} (W_i-w_0)\}^2 K\bigg{(}\frac{W_i-w_0}{h}\bigg{)}} \\
%{U_{L,trans}^{SRD,Y}(\alpha_L^{Y},\beta_L^{Y};G,S) = \sum_{i=1}^n I(W_i < w_0)\{Y_{trans,i}(O_i;G,S) - \alpha_L^{(Y)} - \beta_L^{(Y)} (W_i-w_0)\}^2 K\bigg{(}\frac{W_i-w_0}{h}\bigg{)}.} 
%\end{gather*}
\indent {Let %$\{\hat{\alpha}_{R,IPCW}^{FRD,Y}(G), \hat{\beta}_{R,IPCW}^{FRD,Y}(G), \hat{\alpha}_{L,IPCW}^{FRD,Y}(G), \hat{\beta}_{L,IPCW}^{FRD,Y}(G)\}$ and 
$\{\hat{\alpha}_{R,DR}^{FRD,Y}(G,S)$, $\hat{\beta}_{R,DR}^{FRD,Y}(G,S)$, $\hat{\alpha}_{L,DR}^{FRD,Y}(G,S)$, $\hat{\beta}_{L,DR}^{FRD,Y}(G,S)\}$ be estimators using DR transformation in fuzzy RD design, respectively. Furthermore, we define 
$\{\hat{\alpha}_R^{FRD,Z}$, $\hat{\beta}_R^{FRD,Z}$, $\hat{\alpha}_L^{FRD,Z}$, $\hat{\beta}_L^{FRD,Z}\}$ for the estimators of modeling treatment assignment. We can similarly define estimators for sharp RD designs, say %$\{\hat{\alpha}_{R,IPCW}^{SRD,Y}(G),\hat{\beta}_{R,IPCW}^{SRD,Y}(G), \hat{\alpha}_{L,IPCW}^{SRD,Y}(G), \hat{\beta}_{L,IPCW}^{SRD,Y}(G)\}$ and 
$\{\hat{\alpha}_{R,DR}^{SRD,Y}(G,S)$, $\hat{\beta}_{R,DR}^{SRD,Y}(G,S)$, $\hat{\alpha}_{L,DR}^{SRD,Y}(G,S)$, $\hat{\beta}_{L,DR}^{SRD,Y}(G,S)\}$.}
%\begin{gather*}
%(\hat{\alpha}_{R,IPCW}^{FRD,Y}(G),\hat{\beta}_{R,IPCW}^{FRD,Y}(G)) = \argmin_{\alpha_R^{Y},\beta_R^{Y}} U_{R,IPCW}^{FRD,Y}(\alpha_R^{Y},\beta_R^{Y};G) \\
%(\hat{\alpha}_{L,IPCW}^{FRD,Y}(G),\hat{\beta}_{L,IPCW}^{FRD,Y}(G)) = \argmin_{\alpha_L^{Y},\beta_L^{Y}} U_{L,IPCW}^{FRD,Y}(\alpha_L^{Y},\beta_L^{Y};G) \\
%(\hat{\alpha}_{R,DR}^{FRD,Y}(G,S),\hat{\beta}_{R,DR}^{FRD,Y}(G,S)) = \argmin_{\alpha_R^{Y},\beta_R^{Y}} U_{R,DR}^{FRD,Y}(\alpha_R^{Y},\beta_R^{Y};G,S) \\
%(\hat{\alpha}_{L,DR}^{FRD,Y}(G,S),\hat{\beta}_{L,DR}^{FRD,Y}(G,S)) = \argmin_{\alpha_L^{Y},\beta_L^{Y}} U_{L,DR}^{FRD,Y}(\alpha_L^{Y},\beta_L^{Y};G,S) \\
%(\hat{\alpha}_R^{FRD,Z},\hat{\beta}_R^{FRD,Z}) = \argmin_{\alpha_R^Z,\beta_R^Z} U_R^{FRD,Z}(\alpha_R^Z,\beta_R^Z) \quad (\hat{\alpha}_L^Z,\hat{\beta}_L^Z) = \argmin_{\alpha_L^Z,\beta_L^Z} U_L^{FRD,Z}(\alpha_L^Z,\beta_L^Z), 
%\end{gather*}
%and
%\begin{gather*}
%(\hat{\alpha}_{R,IPCW}^{SRD,Y}(G),\hat{\beta}_{R,IPCW}^{SRD,Y}(G)) = \argmin U_{R,IPCW}^{SRD,Y}(\alpha_R^{Y},\beta_R^{Y};G) \\
%(\hat{\alpha}_{L,IPCW}^{SRD,Y}(G),\hat{\beta}_{L,IPCW}^{SRD,Y}(G)) = \argmin U_{L,IPCW}^{SRD,Y}(\alpha_L^{Y},\beta_L^{Y};G) \\
%(\hat{\alpha}_{R,DR}^{SRD,Y}(G,S),\hat{\beta}_{R,DR}^{SRD,Y}(G,S)) = \argmin U_{R,DR}^{SRD,Y}(\alpha_R^{Y},\beta_R^{Y};G,S) \\
%(\hat{\alpha}_{L,DR}^{SRD,Y}(G,S),\hat{\beta}_{L,DR}^{SRD,Y}(G,S)) = \argmin U_{L,DR}^{SRD,Y}(\alpha_L^{Y},\beta_L^{Y};G,S). 
%\end{gather*}
Then we can derive estimators for the fuzzy and sharp RD designs: 
\begin{gather*}
%\hat{\tau}_{FRD}^{IPCW}(G) = \frac{\hat{\alpha}_{R,IPCW}^{FRD,Y}(G) - \hat{\alpha}_{L,IPCW}^{FRD,Y}(G)}{\hat{\alpha}_R^{FRD,Z} - \hat{\alpha}_L^{FRD,Z}} \quad 
\hat{\tau}_{FRD}^{DR}(G,S) = \frac{\hat{\alpha}_{R,DR}^{FRD,Y}(G) - \hat{\alpha}_{L,DR}^{FRD,Y}(G)}{\hat{\alpha}_R^{FRD,Z} - \hat{\alpha}_L^{FRD,Z}} \\
%\hat{\tau}_{SRD}^{IPCW}(G) = \hat{\alpha}_{R,IPCW}^{SRD,Y}(G) - \hat{\alpha}_{L,IPCW}^{SRD,Y}(G) \quad 
\hat{\tau}_{SRD}^{DR}(G,S) = \hat{\alpha}_{R,DR}^{SRD,Y}(G,S) - \hat{\alpha}_{L,DR}^{SRD,Y}(G,S).
\end{gather*} 
Note that we removed dependence on the bandwidth in the definition of these estimators; we thus assume that it is given. We will discuss how to estimate bandwidth formally in Section 4.3. In the next set of results, we show that under certain conditions, we can prove asymptotic convergence results for the sharp and fuzzy RD estimators. As can be seen, our estimators depend on $G$ and $S$. Let $G_0$ and $S_0$ be true distributions of failure and censoring times, respectively. Let $\hat{G}$ and $\hat{S}$ be estimated distributions of failure and censoring times, respectively. In the estimation, we correctly estimate censoring distribution while we may incorrectly estimate survival distribution. As discussed in the Supplementary Materials, we assume uniform consistency of $\hat{G}$ to $G_0$, and uniform consistency of $\hat{S}$ to $S^*$ where $S^*$ is possibly incorrect model of $S$. We will discuss the estimation of these two distributions in the next subsection. For these two theoretical results, the IPCW and DR estimators are asymptotically normal with some bias.
\begin{thm}
Assume that conditions (C1)-(C5), (R1)-(R9) hold. By Lemma 1-6 in the Supplementary Materials,  
\begin{gather*}
n^{2/5}(\hat{\tau}_{FRD}^{IPCW}(\hat{G}) - \tau_{FRD}- \varphi_{FRD}) \overset{d}{\longrightarrow} N(0,\Sigma_{FRD}^{IPCW}(G_0)) \\
n^{2/5}(\hat{\tau}_{FRD}^{DR}(\hat{G},\hat{S}) - \tau_{FRD}- \varphi_{FRD}) \overset{d}{\longrightarrow} N(0,\Sigma_{FRD}^{DR}(G_0,S^*)) 
%n^{2/5}(\hat{\tau}_{FRD}^{DR}(G,S_0) - \tau_{FRD}- \varphi_{FRD}) \overset{d}{\longrightarrow} N(0,\Sigma_{FRD}^{DR}(G,S_0)) \\
%n^{2/5}(\hat{\tau}_{FRD}^{DR}(G_0,S_0) - \tau_{FRD}- \varphi_{FRD}) \overset{d}{\longrightarrow} N(0,\Sigma_{FRD}^{DR}(G_0,S_0)), 
\end{gather*}
where $\varphi_{FRD}$, $\Sigma_{FRD}^{IPCW}(G)$ and $\Sigma_{FRD}^{DR}(G,S)$ are defined in the Supplementary Materials.
\end{thm}
\noindent For the sharp RD case, the result follows easily from Theorem 1 because there is no need to model $Z|W$ for fuzzy RD. 
\begin{corollary}
Suppose that conditions (C1)-(C5) and (R1)-(R9) in the Supplementary Materials hold. By Lemma 1-6 and Theorem 1 in the Supplementary Materials,  
\begin{gather*}
n^{2/5}(\hat{\tau}_{SRD}^{IPCW}(\hat{G}) - \tau_{SRD}- \varphi_{SRD}) \overset{d}{\longrightarrow} N(0,\Sigma_{SRD}^{IPCW}(G_0)) \\
n^{2/5}(\hat{\tau}_{SRD}^{DR}(\hat{G},\hat{S}) - \tau_{SRD}- \varphi_{SRD}) \overset{d}{\longrightarrow} N(0,\Sigma_{SRD}^{DR}(G_0,S^*)) 
%n^{2/5}(\hat{\tau}_{SRD}^{DR}(G,S_0) - \tau_{SRD}- \varphi_{SRD}) \overset{d}{\longrightarrow} N(0,\Sigma_{SRD}^{DR}(G,S_0)) \\
%n^{2/5}(\hat{\tau}_{SRD}^{DR}(G_0,S_0) - \tau_{SRD}- \varphi_{SRD}) \overset{d}{\longrightarrow} N(0,\Sigma_{SRD}^{DR}(G_0,S_0)), 
\end{gather*}
where $\varphi_{SRD}$, $\Sigma_{SRD}^{IPCW}(G)$ and $\Sigma_{SRD}^{DR}(G,S)$ are defined in the Supplementary Materials.
\end{corollary} 

\indent Now we demonstrate an efficiency result similar to that given in Steingrimsson, Diao, and Strawderman (2019) for a separate problem.  For the sharp RD estimator, due to the `local randomization' result, the DR estimator with true censoring and failure time distributions is more efficient estimator than the estimators with IPCW and DR transformations, which only involve true censoring distribution. For the fuzzy RD estimator, to obtain a similar result, we should account for correlation between the transformed responses and $Z$. Suppose that 
\begin{gather*}
\tau^Y = \lim_{w \downarrow w_0} E(Y|W=w) - \lim_{w \uparrow w_0} E(Y|W=w) \quad \tau^Z = \lim_{w \downarrow w_0} P(Z=1|W=w) - \lim_{w \uparrow w_0} P(Z=1|W=w)
\end{gather*}
We need two conditions to hold: i) $\tau^Y$ and $\tau^Z$ are positive ii) $Cov(Y_{DR}(O;G_0,S_0),Z)$ is smaller than $Cov(Y_{IPCW}(O;G_0),Z)$ and $Cov(Y_{DR}(O;G_0,S),Z)$ at the boundary points. Let 
\begin{gather*}
\hspace*{-10mm}
\eta_{DR}^+(w_0;G,S) = \lim_{w \downarrow w_0} Cov(Y_{DR}(O;G,S), Z|W=w) \quad \eta_{DR}^-(w_0;G,S) = \lim_{w \uparrow w_0} Cov(Y_{DR}(O;G,S), Z|W=w) \\
\hspace*{-10mm}
\eta_{IPCW}^+(w_0;G) = \lim_{w \downarrow w_0} Cov(Y_{IPCW}(O;G), Z|W=w) \quad \eta_{IPCW}^-(w_0;G) = \lim_{w \uparrow w_0} Cov(Y_{IPCW}(O;G), Z|W=w).
\end{gather*}
We formally state these in the theorem below:    
\begin{thm}
%$\hat{\tau}_{SRD}^{DR}(\mathcal{O};G_0,S_0)$ is the most efficient estimator among $\hat{\tau}_{SRD}^{IPCW}(\mathcal{O};G_0)$, $\hat{\tau}_{SRD}^{DR}(\mathcal{O};G_0,S)$,$\hat{\tau}_{SRD}^{DR}(\mathcal{O};G,S_0)$ and $\hat{\tau}_{SRD}^{DR}(\mathcal{O};G_0,S_0)$. When strong ignorability assumption holds, $\hat{\tau}_{FRD}^{DR}(\mathcal{O};G_0,S_0)$ is the most efficient estimator among $\hat{\tau}_{FRD}^{IPCW}(\mathcal{O};G_0)$, $\hat{\tau}_{FRD}^{DR}(\mathcal{O};G_0,S)$,$\hat{\tau}_{FRD}^{DR}(\mathcal{O};G,S_0)$ and $\hat{\tau}_{FRD}^{DR}(\mathcal{O};G_0,S_0)$. 
Suppose that conditions (C1)-(C5) and (R1)-(R9) in the Supplementary Materials hold. Let $AVar$ denote asymptotic variance. Then for the sharp RD estimator,
\begin{gather*}
AVar(\hat{\tau}_{SRD}^{DR}(G_0,S_0)) \leq \min\{AVar(\hat{\tau}_{SRD}^{IPCW}(G_0)),AVar(\hat{\tau}_{SRD}^{DR}(G_0,S))\}.
\end{gather*}
For fuzzy RD estimator, suppose that $\tau^Y > 0$ and $\tau^Z > 0$, and if 
\begin{gather*}
\eta_{DR}^{+}(w_0;G_0,S_0) \leq \min\{\eta_{IPCW}^{+}(w_0;G_0),\eta_{DR}^{+}(w_0;G_0,S)\} \\
\eta_{DR}^{-}(w_0;G_0,S_0) \leq \min\{\eta_{IPCW}^{-}(w_0;G_0),\eta_{DR}^{-}(w_0;G_0,S)\},
\end{gather*} 
then
\begin{gather*}
AVar(\hat{\tau}_{FRD}^{DR}(G_0,S_0)) \leq \min\{AVar(\hat{\tau}_{FRD}^{IPCW}(G_0)),AVar(\hat{\tau}_{FRD}^{DR}(G_0,S))\}.
\end{gather*}
\end{thm}

\subsection{Bandwidth selection}
\indent Once we estimate the censoring unbiased transformation with respect to failure time, we can then apply the existing methodology for RD designs with censored data. The estimated transformation is
\begin{gather*}
\hat{Y}_{IPCW,i}(O_i;\hat{G}) = \frac{\Delta_i \tilde{Y}_i}{\hat{G}(\tilde{T}_i)} \\
\hat{Y}_{DR,i}(O_i;\hat{G},\hat{S}) = \frac{\Delta_i \tilde{Y}_i}{\hat{G}(\tilde{T}_i)} +  \int_0^{\tilde{T}_i} \frac{\hat{Q}_Y(u,W_i)}{\hat{G}(u)}d\hat{M}_{G,i}(u),
\end{gather*}
where $\hat{Q}_Y(\cdot|\cdot)$ is estimator of $Q_Y(\cdot|\cdot)$ and 
\begin{gather*}
\hat{M}_{G,i}(u|W_i) =  I(\tilde{T}_i \leq u, \Delta_i=0) - \int_0^{u} I(\tilde{T}_i \geq s)d\hat{\Lambda}_G(s),
\end{gather*}
and $\hat{\Lambda}_G(\cdot)$ is Nelson-Aalen estimator for $G$. Details of computation of $\hat{G}$ and $\hat{Q}_Y(\cdot,\cdot)$ is found in the Supplementary Materials.\\
\indent In standard nonparametric regression, one important issue is bandwidth selection. 
Ludwig and Miller (2005) and Ludwig and Miller (2007) propose a mean squared error (MSE) type cross-validation criterion. Let $\hat{a}_L(W,\xi,L)$ be $\xi$ quantile of the empirical distribution of $W$ using observations $W_i < w_0$ and let $\hat{a}_R(W,1-\xi)$ be $1-\xi$ quantile of the empirical distribution of $W$ using observations $W_i \geq w_0$. Moreover, let $\hat{\alpha}_L^{(Y)}$ and $\hat{\alpha}_R^{(Y)}$ be estimated parameters for $\alpha_L$ and $\alpha_R$. Criterion from Ludwig and Miller (2005) and Ludwig and Miller (2007) for uncensored data is    
\begin{gather*}
\frac{1}{n} \sum_{\hat{a}_L(W,\xi) \leq W_i \leq \hat{a}_R(W,1-\xi)}(Y_i - \hat{\gamma}(W_i))^2,   
\end{gather*}
where 
\[
				\hat{\gamma}(w) =
				\begin{cases}
				\hat{\alpha}_L(w) & \text{if} \ w < w_0\\
				\hat{\alpha}_R(w) & \text{if} \ w \geq w_0
				\end{cases}
				.
				\]
This criterion still works for the unbiased censoring transformation because it uses a MSE-type criterion and does not depend on variance of $Y$. We now modify the proposal of Ludwig and Miller (2005) and Ludwig and Miller (2007) (henceforth LM) for censored data. For the sharp RD estimator, we consider
\begin{gather*}
\hspace*{-10mm}
(\hat{\alpha}_R^{Y}(w;\hat{G},\hat{S}),\hat{\beta}_R^{Y}(w;\hat{G},\hat{S})) = \argmin \sum_{i=1}^n I(W_i \geq w)\{Y_{DR,i}(O_i;\hat{G},\hat{S}) - \alpha_R^{Y} - \beta_R^{Y} (W_i-w)\}^2 K\bigg{(}\frac{W_i-w_0}{h}\bigg{)} \\
\hspace*{-10mm}
(\hat{\alpha}_L^{Y}(w;\hat{G},\hat{S}),\hat{\beta}_L^{Y}(w;\hat{G},\hat{S})) = \argmin \sum_{i=1}^n I(W_i < w)\{Y_{DR,i}(O_i;\hat{G},\hat{S}) - \alpha_L^{Y} - \beta_L^{Y} (W_i-w)\}^2 K\bigg{(}\frac{W_i-w_0}{h}\bigg{)}. 
\end{gather*}
Then the LM criterion for sharp RD estimator in censored data is given by 
\begin{gather*}
CV_{Y_{DR}}(h;\hat{G},\hat{S}) = \frac{1}{n} \sum_{\hat{a}_L(W,\xi) \leq W_i \leq \hat{a}_R(W,1-\xi)}(\hat{Y}_{DR,i}(O_i;\hat{G},\hat{S}) - \hat{\gamma}_{DR}^Y(W_i))^2, 
\end{gather*}
where 
\[
				\hat{\gamma}_{DR}^Y(w) =
				\begin{cases}
				\hat{\alpha}_{L,DR}^{Y}(w;\hat{G},\hat{S}) & \text{if} \ w < w_0\\
				\hat{\alpha}_{R,DR}^{Y}(w;\hat{G},\hat{S}) & \text{if} \ w \geq w_0 
				\end{cases}
				.
				\]
\noindent We then choose $\hat{h}_{DR}(\hat{G},\hat{S}) = \argmin\limits_{h} CV_{Y_{DR}}(h;\hat{G},\hat{S})$. We can derive a similar quantity for $\hat{Y}_{IPCW,i}(O_i;\hat{G}),i=1,\ldots,n$. For the fuzzy RD estimator, we define
\begin{gather*}
(\hat{\alpha}_R^Z(w),\hat{\beta}_R^Z(w)) = \argmin_{\alpha_R^Z,\beta_R^Z} \sum_{i=1}^n I(W_i \geq w)\{Z_i - \alpha_R^Z - \beta_R^Z (W_i-w)\}^2 K\bigg{(}\frac{W_i-w}{h}\bigg{)} \\
(\hat{\alpha}_L^Z(w),\hat{\beta}_L^Z(w)) = \argmin_{\alpha_L^Z,\beta_L^Z} \sum_{i=1}^n I(W_i < w)\{Z_i - \alpha_L^Z - \beta_L^Z (W_i-w)\}^2 K\bigg{(}\frac{W_i-w}{h}\bigg{)}. 
\end{gather*}
Then we obtain  
\begin{gather*}
CV_{Z}(h) = \frac{1}{n} \sum_{\hat{a}_L(W,\xi) \leq W_i \leq \hat{a}_R(W,1-\xi)}(Z_i - \hat{\gamma}_{Z}(W_i))^2,   
\end{gather*}
where 
\[
				\hat{\gamma}_{Z}(w) =
				\begin{cases}
				\hat{\alpha}_{L}^{Z}(w) & \text{if} \ w < w_0\\
				\hat{\alpha}_{R}^{Z}(w) & \text{if} \ w \geq w_0
				\end{cases}
				.
	\]
Then we obtain $\hat{h}_Z = \argmin\limits_{h} CV_{Z}(h) $. A smaller bandwidth is preferable to reduce the bias of the estimator. Hence for the fuzzy RD estimator, we consider $\min \{\hat{h}_{IPCW}(\hat{G}),\hat{h}_Z\}$ (IPCW) and $\min \{\hat{h}_{DR}(\hat{G},\hat{S}), \hat{h}_{Z}\}$ (DR).%START Here!!
\subsection{Variance estimation}
With the estimation of the censoring and failure time distributions and bandwidth selection, we obtain
\begin{gather*}
%(\hat{\alpha}_{R,IPCW,\hat{h}}^{FRD,Y}(\hat{G}),\hat{\beta}_{R,IPCW,\hat{h}}^{FRD,Y}(\hat{G})) = \argmin_{\alpha_R^{Y},\beta_R^{Y}} U_{R,IPCW,\hat{h}}^{FRD,Y}(\alpha_R^{Y},\beta_R^{Y};\hat{G}) \\
%(\hat{\alpha}_{L,IPCW,\hat{h}}^{FRD,Y}(\hat{G}),\hat{\beta}_{L,IPCW,\hat{h}}^{FRD,Y}(\hat{G})) = \argmin_{\alpha_L^{Y},\beta_L^{Y}} U_{L,IPCW,\hat{h}}^{FRD,Y}(\alpha_L^{Y},\beta_L^{Y};\hat{G}) \\
%(\hat{\alpha}_{R,IPCW,\hat{h}}^{FRD,Y}(\hat{G}),\hat{\beta}_{R,IPCW,\hat{h}}^{FRD,Y}(\hat{G})) = \argmin_{\alpha_R^{Y},\beta_R^{Y}} U_{R,IPCW,\hat{h}}^{FRD,Y}(\alpha_R^{Y},\beta_R^{Y};\hat{G}) \\
%(\hat{\alpha}_{L,IPCW,\hat{h}}^{FRD,Y}(\hat{G}),\hat{\beta}_{L,IPCW,\hat{h}}^{FRD,Y}(\hat{G})) = \argmin_{\alpha_L^{Y},\beta_L^{Y}} U_{L,IPCW,\hat{h}}^{FRD,Y}(\alpha_L^{Y},\beta_L^{Y};\hat{G}) \\
(\hat{\alpha}_{R,DR,\hat{h}}^{FRD,Y}(\hat{G},\hat{S}),\hat{\beta}_{R,DR,\hat{h}}^{Y}(\hat{G},\hat{S})) = \argmin_{\alpha_R^{Y},\beta_R^{Y}} U_{R,DR,\hat{h}}^{FRD,Y}(\alpha_R^{Y},\beta_R^{Y};\hat{G},\hat{S}) \\
(\hat{\alpha}_{L,DR,\hat{h}}^{FRD,Y}(\hat{G},\hat{S}),\hat{\beta}_{L,DR,\hat{h}}^{Y}(\hat{G},\hat{S})) = \argmin_{\alpha_L^{Y},\beta_L^{Y}} U_{L,DR,\hat{h}}^{FRD,Y}(\alpha_L^{Y},\beta_L^{Y};\hat{G},\hat{S}) \\
(\hat{\alpha}_{R,\hat{h}}^{FRD,Z},\hat{\beta}_{R,\hat{h}}^{FRD,Z}) = \argmin_{\alpha_R^Z,\beta_R^Z} U_{R,\hat{h}}^{FRD,Z}(\alpha_R^Z,\beta_R^Z) \quad (\hat{\alpha}_L^Z,\hat{\beta}_L^Z) = \argmin_{\alpha_L^Z,\beta_L^Z} U_{L,\hat{h}}^{FRD,Z}(\alpha_L^Z,\beta_L^Z), 
\end{gather*}
%where $U_{R,IPCW,\hat{h}}^{FRD,Y}(\cdot)$, $U_{R,DR,\hat{h}}^{FRD,Y}(\cdot,\cdot)$ and $U_{R,\hat{h}}^{FRD,Z}$ are $U_{R,IPCW}^{FRD,Y}(\cdot)$, $U_{R,DR}^{FRD,Y}(\cdot,\cdot)$ and $U_{R}^{FRD,Z}$ with estimated bandwidth. We can define similarly for estimation functions with $W < w_0$. 
{where $U_{R,DR,\hat{h}}^{FRD,Y}$ and $U_{R,\hat{h}}^{FRD,Z}$ correspond to $U_{R,DR}^{FRD,Y}$ and $U_{R}^{FRD,Z}$ with an estimated bandwidth. We can define similarly for estimation functions with IPCW transformation. Estimators for sharp RD designs can be similarly defined.} Hence the proposed fuzzy RD and sharp RD estimators based on $\hat{G}$ and $\hat{S}$ are 
\begin{gather*}
%\hat{\tau}_{FRD,\hat{h}}^{IPCW}(\hat{G}) = \frac{\hat{\alpha}_{R,IPCW,\hat{h}}^{FRD,Y}(\hat{G}) - \hat{\alpha}_{L,IPCW,\hat{h}}^{FRD,Y}(\hat{G})}{\hat{\alpha}_{R,\hat{h}}^Z - \hat{\alpha}_{L,\hat{h}}^Z} \\
%\hat{\tau}_{SRD,\hat{h}}^{IPCW}(\hat{G}) = \hat{\alpha}_{R,IPCW,\hat{h}}^{SRD,Y}(\hat{G}) - \hat{\alpha}_{L,IPCW,\hat{h}}^{SRD,Y}(\hat{G}) \\
\hat{\tau}_{FRD,\hat{h}}^{DR}(\hat{G},\hat{S}) = \frac{\hat{\alpha}_{R,DR,\hat{h}}^{FRD,Y}(\hat{G},\hat{S}) - \hat{\alpha}_{L,DR,\hat{h}}^{FRD,Y}(\hat{G},\hat{S})}{\hat{\alpha}_{R,\hat{h}}^Z - \hat{\alpha}_{L,\hat{h}}^Z} \\
\hat{\tau}_{SRD,\hat{h}}^{DR}(\hat{G},\hat{S}) = \hat{\alpha}_{R,DR,\hat{h}}^{SRD,Y}(\hat{G},\hat{S}) - \hat{\alpha}_{L,DR,\hat{h}}^{SRD,Y}(\hat{G},\hat{S}).
\end{gather*}
{where %$\{\hat{\alpha}_{R,IPCW,\hat{h}}^{SRD,Y}(\hat{G})$, $\hat{\alpha}_{L,IPCW,\hat{h}}^{SRD,Y}(\hat{G})\}$ and 
$\{\hat{\alpha}_{R DR,\hat{h}}^{SRD,Y}(\hat{G},\hat{S})$, $\hat{\alpha}_{L,DR,\hat{h}}^{SRD,Y}(\hat{G},\hat{S})\}$ are sharp RD estimators with estimated bandwidth from DR transformation.} For variance estimation, one may use the methods based on the asymptotic results in Section 4.2. Let $e_1 = (1,0)^T$ and $a(u) = (1,u)^T$. Define $g(\cdot)$ to common density of $W_i$. Using the expressions in the Supplementary Materials, the asymptotic variance of $\hat{\tau}_{SRD}^{DR}(\hat{G},\hat{S})$ can be expressed as
\begin{gather*}
\frac{1}{nh} \{g(w_0)\}^{-1}\{e_1^T(\sigma_{DR}^{2+}(w_0;G_0,S^*)\mathbf{\Gamma}^{-1}\boldsymbol{\vartheta}\mathbf{\Gamma}^{-1} + \sigma_{DR}^{2-}(w_0;G_0,S^*)\mathbf{\Gamma}^{-1}\boldsymbol{\vartheta}\mathbf{\Gamma}^{-1})e_1\}[1+o_p(1)], 
\end{gather*}
where 
\begin{gather*} 
\boldsymbol{\vartheta} = \int_0^{\infty} \{K(u)\}^2a(u)a^T(u) \quad \mathbf{\Gamma} = \int_0^{\infty} K(u)a(u)a^T(u) \\
\sigma_{DR}^{2+}(w_0;G,S) = \lim_{\epsilon \downarrow w_0} Var(Y_{DR}(O;G,S)|W=w) \quad \sigma_{DR}^{2-}(w_0;G,S) = \lim_{\epsilon \uparrow w_0} Var(Y_{DR}(O;G,S)|W=w).
\end{gather*}
Let $\mathbf{X}_h$ be $n \times 2$ matrix with the first column being 1 and the second column being $\dfrac{W_i-w_0}{h}, i=1,\ldots,n$. Moreover, let $\mathbf{W}_{h+}$ and $\mathbf{W}_{h-}$ be diagonal matrix with diagonal elements with $I(W_1 \geq w_0)K\bigg{(}\dfrac{W_i - w_0}{h}\bigg{)}$ and $I(W_1 < w_0)K\bigg{(}\dfrac{W_i - w_0}{h}\bigg{)}, i=1\ldots,n$, respectively.
%\begin{gather*}
%\mathbf{X}_h = 
%\begin{pmatrix}
%1 & \frac{W_1-w_0}{h} \\
%\vdots & \vdots \\
%1 & \frac{W_n-w_0}{h} 
%\end{pmatrix}
%\\
%\mathbf{W}_{h+} =
%\begin{pmatrix}
%I(W_1 \geq w_0)K\bigg{(}\dfrac{W_1 - w_0}{h}\bigg{)} & 0 & \ldots & 0 \\
%0 & I(W_2 \geq w_0)K\bigg{(}\dfrac{W_2 - w_0}{h}\bigg{)} & \ldots & 0 \\
%\vdots & \vdots & \ddots & \vdots \\
%0 & 0 & \ldots 0 & I(W_n \geq w_0)K\bigg{(}\dfrac{W_n - w_0}{h}\bigg{)}
%\end{pmatrix}
%\\
%\mathbf{W}_{h-} =
%\begin{pmatrix}
%I(W_1 < w_0)K\bigg{(}\dfrac{W_1 - w_0}{h}\bigg{)} & 0 & \ldots & 0 \\
%0 & I(W_2 < w_0)K\bigg{(}\dfrac{W_2 - w_0}{h}\bigg{)} & \ldots & 0 \\
%\vdots & \vdots & \ddots & \vdots \\
%0 & 0 & \ldots 0 & I(W_n < w_0)K\bigg{(}\dfrac{W_n - w_0}{h}\bigg{)}
%\end{pmatrix}.
%\end{gather*}
Define 
\begin{gather*}
\mathbf{\Gamma}_{h+} = \mathbf{X}_h^T\mathbf{W}_{h+}\mathbf{X}_h \quad \mathbf{\Gamma}_{h-} = \mathbf{X}_h^T\mathbf{W}_{h-}\mathbf{X}_h.
\end{gather*}
By Calonico, Cattaneo and Titiunik (2014),  
\begin{gather*}
\frac{1}{nh} \{g(w_0)\}^{-1}\{e_1^T(\sigma_{DR}^{2+}(w_0;G_0,S^*)\mathbf{\Gamma}^{-1}\boldsymbol{\vartheta}\mathbf{\Gamma}^{-1} + \sigma_{DR}^{2-}(w_0;G_0,S^*)\mathbf{\Gamma}^{-1}\boldsymbol{\vartheta}\mathbf{\Gamma}^{-1})e_1\}[1+o_p(1)]  \\
=  \frac{1}{n}e_1^T(\boldsymbol{\Gamma}_{h+}^{-1}\boldsymbol{\phi}_{YY+,DR}\boldsymbol{\Gamma}_{h+}^{-1} + \boldsymbol{\Gamma}_{h-}^{-1}\boldsymbol{\phi}_{YY-,DR}\boldsymbol{\Gamma}_{h-}^{-1})e_1, 
\end{gather*}
where
\begin{gather*}
\boldsymbol{\phi}_{YY+,DR} = \frac{1}{n}\sum_{i=1}^n I(W_i \geq w_0) K\bigg{(}\frac{W_i-w_0}{h}\bigg{)}K\bigg{(}\frac{W_i-w_0}{h}\bigg{)}b_{i}b_{i}^T\sigma_{DR,(1)}^{2}(W_i;G_0,S^*) \\
\boldsymbol{\phi}_{YY-,DR} = \frac{1}{n}\sum_{i=1}^n I(W_i < w_0) K\bigg{(}\frac{W_i-w_0}{h}\bigg{)}K\bigg{(}\frac{W_i-w_0}{h}\bigg{)}b_{i}b_{i}^T\sigma_{DR,(0)}^{2}(W_i;G_0,S^*),
\end{gather*} 
where $b_i = (1,h^{-1}(W_i-w_0))^T$ and 
\begin{gather*}
\sigma_{DR,(1)}^{2}(w;G_0,S^*) = Var\bigg{\{}\frac{\Delta Y(1)}{G_0(T)} + \int_0^{\infty} \frac{Q_{Y(1)}(u,W;S^*)}{G_0(u)}dM_{G}(u)\bigg{|}W=w\bigg{\}} \\
\sigma_{DR,(0)}^{2}(w;G_0,S^*) = Var\bigg{\{}\frac{\Delta Y(0)}{G_0(T)} + \int_0^{\infty} \frac{Q_{Y(0)}(u,W;S^*)}{G_0(u)}dM_{G}(u)\bigg{|}W=w\bigg{\}}. 
\end{gather*}
where $Q_{Y(k)}(u,W;S^*) = E_{S^*}(Y(k)|W,T \geq u), k=0, 1$. We can obtain a similar result for the IPCW estimator. The first approach is to use plug-in residuals for the sandwich variance estimator.  Let $\hat{\epsilon}_{Y+,DR,i} = \hat{Y}_{DR,i}(O_i;\hat{G},\hat{S}) - \hat{\alpha}_{R,DR,\hat{h}}^{SRD,Y}(\hat{G},\hat{S})$ and $\hat{\epsilon}_{Y-,DR,i} = \hat{Y}_{DR,i}(O_i;\hat{G},\hat{S}) - \hat{\alpha}_{L,DR,\hat{h}}^{SRD,Y}(\hat{G},\hat{S})$. %where $\hat{Y}_i^{trans}(O_i)$ is either $\hat{Y}_{IPCW,i}(O_i;\hat{G})$ or $\hat{Y}_{DR,i}(O_i;\hat{G},\hat{S})$. 
Then one may wish to use
\begin{gather*}
\hat{\boldsymbol{\phi}}_{YY+,DR,\hat{h}}^{pir} = \frac{1}{n}\sum_{i=1}^n I(W_i \geq w_0) K\bigg{(}\frac{W_i-w_0}{\hat{h}}\bigg{)}K\bigg{(}\frac{W_i-w_0}{\hat{h}}\bigg{)}b_{i}b_{i}^T\hat{\epsilon}_{Y+,DR,i}^2 \\
\hat{\boldsymbol{\phi}}_{YY-,DR,\hat{h}}^{pir} = \frac{1}{n}\sum_{i=1}^n I(W_i < w_0) K\bigg{(}\frac{W_i-w_0}{\hat{h}}\bigg{)}K\bigg{(}\frac{W_i-w_0}{\hat{h}}\bigg{)}b_{i}b_{i}^T\hat{\epsilon}_{Y-,DR,i}^2.
\end{gather*}
The second approach is to use a nonparametric nearest neighborhood (NN) variance estimator as in Calonico, Cattaneo and Titiunik (2014). This approach is advantageous in that it does not require nonparametric smoothing and is robust (Abadie and Imbens, 2006). Mimicking the approach from Abadie and Imbens (2006) and Calonico, Cattaneo and Titiunik (2014),
\begin{gather*}
\hat{\boldsymbol{\phi}}_{YY+,DR,\hat{h}}^{rb} = \frac{1}{n}\sum_{i=1}^n I(W_i \geq w_0) K\bigg{(}\frac{W_i-w_0}{\hat{h}}\bigg{)}K\bigg{(}\frac{W_i-w_0}{\hat{h}}\bigg{)}b_{i}b_{i}^T\hat{\sigma}_{YY+,DR}^2(W_i) \\
\hat{\boldsymbol{\phi}}_{YY-,DR,\hat{h}}^{rb} = \frac{1}{n}\sum_{i=1}^n I(W_i < w_0) K\bigg{(}\frac{W_i-w_0}{\hat{h}}\bigg{)}K\bigg{(}\frac{W_i-w_0}{\hat{h}}\bigg{)}b_{i}b_{i}^T\hat{\sigma}_{YY-,DR}^2(W_i),
\end{gather*}
where 
\begin{gather*}
\hat{\sigma}_{YY+,DR}^2(W_i) = I(W_i \geq w_0)\frac{K}{K+1}\bigg{(}\hat{Y}_i^{DR}(O_i;\hat{G},\hat{S}) - \frac{1}{K}\sum_{k=1}^K \hat{Y}_{l_{+,k}(i)}^{DR}(O_i;\hat{G},\hat{S}) \bigg{)}^2 \\
\hat{\sigma}_{YY-,DR}^2(W_i) = I(W_i < w_0)\frac{K}{K+1}\bigg{(}\hat{Y}_i^{trans}(O_i;\hat{G},\hat{S}) - \frac{1}{K}\sum_{k=1}^K \hat{Y}_{l_{-,k}(i)}^{DR}(O_i;\hat{G},\hat{S}) \bigg{)}^2,
\end{gather*}
where $\hat{Y}_{l_{+,k}(i)}^{DR}(O_i;\hat{G},\hat{S})$ is $k$th closest unit to unit $i$ among $\{W_i : W_i \geq w_0\}$ and  $\hat{Y}_{l_{-,k}}^{DR}(O_i;\hat{G},\hat{S})$ is $k$th closest unit to unit $i$ among $\{W_i : W_i < w_0\}$ for variable $\hat{Y}_i^{DR}(O_i;\hat{G},\hat{S})$, respectively. Let $\mathbf{X}_{\hat{h}}$, $\mathbf{W}_{\hat{h}+}$ and $\mathbf{W}_{\hat{h}-}$ are $\mathbf{X}_{h}$, $\mathbf{W}_{h+}$ and $\mathbf{W}_{h-}$ with estimated bandwidth. Then the variance estimator will be 
\begin{gather*}
\frac{1}{n}e_1^T(\hat{\boldsymbol{\Gamma}}_{\hat{h}+}^{-1}\hat{\boldsymbol{\phi}}_{YY+,DR,\hat{h}}\hat{\boldsymbol{\Gamma}}_{\hat{h}+}^{-1} + \hat{\boldsymbol{\Gamma}}_{\hat{h}-}^{-1}\hat{\boldsymbol{\phi}}_{YY-,DR,\hat{h}}\hat{\boldsymbol{\Gamma}}_{\hat{h}-}^{-1})e_1,
\end{gather*}
where    
\begin{gather*}
\hat{\boldsymbol{\Gamma}}_{\hat{h}+} = \mathbf{X}_{\hat{h}}^T\mathbf{W}_{\hat{h}+}\mathbf{X}_{\hat{h}} \quad \hat{\mathbf{\Gamma}}_{\hat{h}-} = \mathbf{X}_{\hat{h}}^T\mathbf{W}_{\hat{h}-}\mathbf{X}_{\hat{h}},
\end{gather*}
and $\hat{\boldsymbol{\phi}}_{YY+,DR,\hat{h}}$ is either $\hat{\boldsymbol{\phi}}_{YY+,DR,\hat{h}}^{pir}$ or $\hat{\boldsymbol{\phi}}_{YY+,DR,\hat{h}}^{rb}$ and define $\hat{\boldsymbol{\phi}}_{YY-,DR,\hat{h}}$ similarly. For variance estimation with the fuzzy RD estimator, by extending the method from the sharp RD design case, we can calculate $\hat{\boldsymbol{\phi}}_{YY+,DR,\hat{h}}^{pir}$, $\hat{\boldsymbol{\phi}}_{YY-,DR,\hat{h}}^{pir}$ $\hat{\boldsymbol{\phi}}_{YY+,DR,\hat{h}}^{rb}$, $\hat{\boldsymbol{\phi}}_{YY-,DR,\hat{h}}^{rb}$. In these calculations, $\hat{\epsilon}_{Y+,DR,i}$ and $\hat{\epsilon}_{Y-,DR,i}$ are based on $\hat{\alpha}_{R,DR,\hat{h}}^{FRD,Y}(\hat{G},\hat{S})$ and $\hat{\alpha}_{L,DR,\hat{h}}^{FRD,Y}(\hat{G},\hat{S})$, resepctively. Now the covariance term between the transformed response and $Z$ should be reflected in the estimation. Let $\hat{\epsilon}_{Z+,i} = Z_i - \hat{\alpha}_{R,\hat{h}}^{FRD,Z}$ and $\hat{\epsilon}_{Z-,i} = Z_i - \hat{\alpha}_{L,\hat{h}}^{FRD,Z}$. Then
\begin{gather*}
\hat{\boldsymbol{\phi}}_{YZ+,DR,\hat{h}}^{rb} = \frac{1}{n}\sum_{i=1}^n I(W_i \geq w_0) K\bigg{(}\frac{W_i-w_0}{\hat{h}}\bigg{)}K\bigg{(}\frac{W_i-w_0}{\hat{h}}\bigg{)}b_{i}b_{i}^T\hat{\sigma}_{YZ+,DR}^2(W_i) \\
\hat{\boldsymbol{\phi}}_{YZ-,DR,\hat{h}}^{rb} = \frac{1}{n}\sum_{i=1}^n I(W_i < w_0) K\bigg{(}\frac{W_i-w_0}{\hat{h}}\bigg{)}K\bigg{(}\frac{W_i-w_0}{\hat{h}}\bigg{)}b_{i}b_{i}^T\hat{\sigma}_{YZ-,DR}^2(W_i) \\
\hat{\boldsymbol{\phi}}_{ZZ+,\hat{h}}^{rb} = \frac{1}{n}\sum_{i=1}^n I(W_i \geq w_0) K\bigg{(}\frac{W_i-w_0}{\hat{h}}\bigg{)}K\bigg{(}\frac{W_i-w_0}{\hat{h}}\bigg{)}b_{i}b_{i}^T\hat{\sigma}_{ZZ+}^2(W_i) \\
\hat{\boldsymbol{\phi}}_{ZZ-,\hat{h}}^{rb} = \frac{1}{n}\sum_{i=1}^n I(W_i < w_0) K\bigg{(}\frac{W_i-w_0}{\hat{h}}\bigg{)}K\bigg{(}\frac{W_i-w_0}{\hat{h}}\bigg{)}b_{i}b_{i}^T\hat{\sigma}_{ZZ-}^2(W_i),
\end{gather*}
where 
\begin{gather*}
\hat{\sigma}_{YZ+,DR}^2(W_i) = I(W_i \geq w_0)\frac{K}{K+1}\bigg{(}\hat{Y}_i^{DR}(O_i) - \frac{1}{K}\sum_{k=1}^K \hat{Y}_{l_{+,k}(i)}^{DR}(O_i) \bigg{)}\bigg{(}Z_i - \frac{1}{K}\sum_{k=1}^K Z_{l_{+,k}(i)} \bigg{)} \\
\hat{\sigma}_{YZ-,DR}^2(W_i) = I(W_i < w_0)\frac{K}{K+1}\bigg{(}\hat{Y}_i^{DR}(O_i) - \frac{1}{M}\sum_{k=1}^K \hat{Y}_{l_{-,k}(i)}^{DR}(O_i) \bigg{)}\bigg{(}Z_i - \frac{1}{K}\sum_{k=1}^K Z_{l_{-,k}(i)} \bigg{)} \\ 
\hat{\sigma}_{ZZ+}^2(W_i) = I(W_i \geq w_0)\frac{K}{K+1}\bigg{(}Z_i - \frac{1}{K}\sum_{k=1}^K Z_{l_{+,k}(i)} \bigg{)}^2 \\
\hat{\sigma}_{ZZ-}^2(W_i) = I(W_i < w_0)\frac{K}{K+1}\bigg{(}Z_i - \frac{1}{K}\sum_{k=1}^K Z_{l_{-,k}(i)} \bigg{)}^2.
\end{gather*}
where $Z_{l_{+,k}(i)}$ and $Z_{l_{-,k}(i)}$ are defined similarly as $\hat{Y}_{l_{+,k}(i)}^{DR}(O_i)$ and $\hat{Y}_{l_{-,k}(i)}^{DR}(O_i)$. For plug-in approach, in addtion to quantities from sharp RD, we can calculate $\{\hat{\boldsymbol{\phi}}_{YZ+,DR,\hat{h}}^{pir}, \hat{\boldsymbol{\phi}}_{YZ-,DR,\hat{h}}^{pir}, \hat{\boldsymbol{\phi}}_{ZZ+,\hat{h}}^{pir}$,$\hat{\boldsymbol{\phi}}_{ZZ-,\hat{h}}^{pir}$\} similarly by replacing $\{\hat{\sigma}_{YZ+,DR}^2(W_i)$, $\hat{\sigma}_{YZ-,DR}^2(W_i),\hat{\sigma}_{ZZ+}^2(W_i)$, $\hat{\sigma}_{ZZ-}^2(W_i)\}$ by $\{\hat{\epsilon}_{Y+,DR,i}\hat{\epsilon}_{Z+,i}, \hat{\epsilon}_{Y-,DR,i}\hat{\epsilon}_{Z-,i}, \hat{\epsilon}_{Z+,i}^2, \hat{\epsilon}_{Z-,i}^2$\} in fuzzy RD. \\
\indent Let $\hat{\tau}_{DR}^{Y}(\hat{G},\hat{S}) = \hat{\alpha}_{R,DR,\hat{h}}^{FRD,Y}(\hat{G},\hat{S}) - \hat{\alpha}_{L,DR,\hat{h}}^{FRD,Y}(\hat{G},\hat{S})$ and $\hat{\tau}_{Z} = \hat{\alpha}_{R,\hat{h}}^{FRD,Z} - \hat{\alpha}_{L,\hat{h}}^{FRD,Z}$. Denote $\hat{\tau}_{DR}^{Y}$ to be $\hat{\tau}_{DR}^{Y}(\hat{G},\hat{S})$. Moreover, $\hat{\boldsymbol{\phi}}_{ZZ+,\hat{h}}$ is either $\hat{\boldsymbol{\phi}}_{ZZ+,\hat{h}}^{pir}$ or $\hat{\boldsymbol{\phi}}_{ZZ+,\hat{h}}^{rb}$. $\hat{\boldsymbol{\phi}}_{ZZ-,\hat{h}}$ is defined similarly. By the plug-in approach from asymptotic results (see Supplementary Materials),
\begin{gather*}
\frac{1}{\hat{\tau}_{Z}^2}(\hat{V}_{YY+,DR,\hat{h}} + \hat{V}_{YY-,DR,\hat{h}}) - \frac{2\hat{\tau}_{DR}^{Y}}{\hat{\tau}_{Z}^3}(\hat{V}_{YZ+,DR,\hat{h}}+\hat{V}_{YZ-,DR,\hat{h}}) + \frac{(\hat{\tau}_{DR}^{Y})^2}{\hat{\tau}_{Z}^4}(\hat{V}_{ZZ+,\hat{h}}+\hat{V}_{ZZ-,\hat{h}}),
\end{gather*}
where 
\begin{gather*}
\hat{V}_{YY+,DR,\hat{h}} = \frac{1}{n}e_1^T\hat{\boldsymbol{\Gamma}}_{\hat{h}+}^{-1}\hat{\boldsymbol{\phi}}_{YY+,DR,\hat{h}}\hat{\boldsymbol{\Gamma}}_{\hat{h}+}^{-1}e_1 \quad \hat{V}_{YY-,DR,\hat{h}} = \frac{1}{n}e_1^T\hat{\boldsymbol{\Gamma}}_{\hat{h}-}^{-1}\hat{\boldsymbol{\phi}}_{YY-,DR,\hat{h}}\hat{\boldsymbol{\Gamma}}_{\hat{h}-}^{-1}e_1 \\
\hat{V}_{YZ+,DR,\hat{h}} = \frac{1}{n}e_1^T\hat{\Gamma}_{\hat{h}+}^{-1}\hat{\boldsymbol{\phi}}_{YZ+,DR,\hat{h}}\hat{\boldsymbol{\Gamma}}_{\hat{h}+}^{-1}e_1 \quad \hat{V}_{YZ-,DR,\hat{h}} = \frac{1}{n}e_1^T\hat{\boldsymbol{\Gamma}}_{\hat{h}-}^{-1}\hat{\boldsymbol{\phi}}_{YZ-,DR,\hat{h}}\hat{\boldsymbol{\Gamma}}_{\hat{h}-}^{-1}e_1 \\
\hat{V}_{ZZ+,\hat{h}} = \frac{1}{n}e_1^T\hat{\boldsymbol{\Gamma}}_{\hat{h}+}^{-1}\hat{\boldsymbol{\phi}}_{ZZ+,\hat{h}}\hat{\boldsymbol{\Gamma}}_{\hat{h}+}^{-1}e_1 \quad \hat{V}_{ZZ-} = \frac{1}{n}e_1^T\hat{\boldsymbol{\Gamma}}_{\hat{h}-}^{-1}\hat{\boldsymbol{\phi}}_{ZZ-,\hat{h}}\hat{\boldsymbol{\Gamma}}_{\hat{h}-}^{-1}e_1. 
\end{gather*}
The nonparametric bootstrap is another method we will consider for standard error estimation. %The procedure is as follows:
%\begin{enumerate}
%\item Bootstrap the data. Let $(\tilde{Y}_i^{(b)},\Delta_i^{(b)},W_i^{(b)},Z_i^{(b)}),i=1,\ldots,n$ be bootstrapped data.
%\item Compute the IPCW and DR transformations based on the bootstrapped data.
%\item Estimate fuzzy RD and sharp RD estimators using the new IPCW and DR transformations.
%\item Repeat Step 1-3 $B$ times.
%\item Compute the variance from the empirical distribution of the bootstrapped fuzzy RD and sharp RD estimators.
%\end{enumerate}
For implementation of our method, we can use existing software for uncensored data. The R package \verb|rdboust| (Calonico, Cattaneo and Titiunik, 2015b) is a powerful tool to perform statistical inference for RD designs. To implement the proposed methods, one can simply transform the response by methods in Section 4.3. Then with the transformed quantities as a new response, we can estimate regression coefficients along with standard errors.  

\section{Simulation results} 
We performed simulation studies to evaluate the finite-sample properties of our proposed estimators. The forcing variable $W$ is generated as a $Unif(0,1)$ random variable. The error variable is generated as $\epsilon \sim N(0,0.5)$. Regression coefficients are set to be $\beta_{10}=2$, $\beta_{20}=1$ and $\beta_{30}=1$. The response is generated from the following model: 
\begin{gather*}
T=\exp(\beta_{10} + \beta_{20}W + \beta_{30}I(W \geq 0.5) + \epsilon).
\end{gather*}
Censoring is generated as a $Unif(0,50)$ random variable that is independent of $T$ and $W$. Three conditional expectation methods were considered in the simulation study: Cox model, Log-normal and Log-Logistic model. We use Kaplan-Meier to estimate $G$. To ensure positivity of $\hat{G}$, we truncate $\tilde{T}$ by time point $\omega$ where $\omega$ is 95th percentile of observed time for estimation of $G$ (Steingrimsson et al. 2016). Censoring distribution is fitted using Kaplan-Meier. Sample sizes are $n=200$ and $n=400$. The number of bootstraps within each simulation is 50. To select bandwidth by cross-validation, it is important to select $\xi$ based on the range of dataset. The value of $\xi$ is 0.5.   The amount of observed censoring is approximately 51\% across the simulations. Kernel function is triangular function, which is
\begin{gather*}
K(u) = 1 - |u|.
\end{gather*}  
\indent Table 1 shows finite-sample properties of the estimator $\hat{\tau}_{SRD}$. In the columns denoting standard error calculation and coverage, NN, Plug-in and Boot denote the nearest neighborhood, plug-in residual and bootstrap approaches, respectively.  For coverage, all the calculations are based on the normal approximation except the Boot\_ED, which denotes coverage based on the empirical distribution of bootstrapped samples. The IPCW approach is more biased than doubly robust approach. Except for the bootstrap, in general, the coverage of the estimators satisfies the 95\% nominal level. The efficiency gain of DR approach compared to IPCW approach is noticeable. The performances of DR approaches across the conditional expectations are very stable. The results from the DR approach confirms the augmentation theory results from Tsiatis (2007). \\
%\indent We also examine finite sample properties of the estimator of survival function for SRD. In our simulation setting, we can conveniently calculate true value of $\tau_{SRD}^S(t)$. Table 2 and Table 3 show that numerical results for $\hat{\tau}_{SRD}^S(t)$. Numerical results show that IPCW estimator and DR estimator provide desired finite sample performance. \\
\indent In the simulation study, we consider the fuzzy RD based on a modification of the simulation setting in Yang (2013). We generate $W \sim Unif(-1,1)$, and let $V = I(W \geq 0)$. Next, we generate $\kappa \sim N(0,0.25)$ and independent of all these aforementioned variables. Then treatment variable is defined as $Z = I(-0.5 + V + W + \kappa > 0)$. Then we generate $\epsilon \sim N(0,0.25)$ that is independent of aforementioned variables. Failure time is defined as $T = \exp(\beta_{10} + \beta_{20} W + \beta_{30} Z + \epsilon)$ where regression coefficients are set to be $\beta_{10}=2$, $\beta_{20}=1$ and $\beta_{30}=1$. The censoring variable is generated as $Unif(0,50)$. The average censoring rate is approximately 39\%. In this case, the denominator for the true value is calculated as 
\begin{gather*}
\lim_{w \downarrow 0} P(Z = 1|W=w) = P(\kappa > -0.5) = 1-  P(\kappa \leq -0.5) = 1 - \Phi(-2) \\
\lim_{w \uparrow 0} P(Z = 1|W=w) = P(\kappa > 0.5) = 1-  P(\kappa \leq 0.5) = 1 - \Phi(2),
\end{gather*}
where $\Phi$ is an inverse function of the standard normal cumulative distribution function. Hence the denominator should be $\Phi(2) - \Phi(-2)$. For the numerator,
\begin{gather*}
\lim_{w \downarrow 0} E\{\log(T)|W=w\} = \beta_{10} + \beta_{20} \times 0 + \lim_{w \downarrow 0} P(Z = 1|W=w) \\
\lim_{w \uparrow 0} E\{\log(T)|W=w\} = \beta_{10} + \beta_{20} \times 0 + \lim_{w \uparrow 0} P(Z = 1|W=w).
\end{gather*}
Hence the numerator and denominator are equal so that the average treatment effect for those who comply with the treatment assignment is 1. We use the same conditional expectation methods as in the sharp RD case. Table 2 shows the numerical results for sample sizes $n=250$ and $n=500$.  For all approaches, the bias is greater than those reported in Table 1.  This makes sense because the estimator in fuzzy RD has a denominator which requires estimation by nonparametric method, which introduces bias. As with the sharp RD situation, the DR method shows good performance regardless of choice of conditional expectation. The IPCW method has a larger bias than the DR methods. 
%The standard error from the bootstrap does not reduce from increase of sample size. 
The coverage probability tend to perform better in larger sample sizes. It is interesting to note that the coverage based on the empirical distribution using the bootstrap works better for fuzzy RD relative to sharp RD.

\begin{table}[!ht]
\centering
\begin{tabular}{rr|r|rrr|rrrr}
  \hline
&& \multirow{2}{*}{Bias} & \multicolumn{3}{c|}{SE} & \multicolumn{4}{c}{Cover} \\
 && &  NN & Plug-in & Boot & NN & Plug-in & Boot & Boot\_ED \\ 
  \hline
$n=200$ & IPCW & 0.074  & 1.422 & 1.390 & 1.542 & 0.900 & 0.894 & 0.916 & 0.912 \\ 
& Cox & -0.004  & 0.123 & 0.119 & 0.121 & 0.942 & 0.944 & 0.942 & 0.908 \\ 
& Log-norm & -0.008  & 0.136 & 0.132 & 0.136 & 0.946 & 0.944 & 0.942 & 0.910 \\ 
& Log-log & -0.008  & 0.136 & 0.132 & 0.139 & 0.942 & 0.942 & 0.946 & 0.912 \\ \hline
$n=400$ & IPCW & 0.091  & 1.008 & 0.995 & 1.084 & 0.930 & 0.930 & 0.944 & 0.924 \\ 
& Cox & -0.003  & 0.086 & 0.085 & 0.084 & 0.940 & 0.932 & 0.938 & 0.898 \\ 
& Log-norm & -0.004  & 0.095 & 0.093 & 0.094 & 0.922 & 0.920 & 0.936 & 0.892 \\ 
& Log-log & -0.004  & 0.096 & 0.094 & 0.096 & 0.920 & 0.924 & 0.936 & 0.890 \\ 
   \hline
\end{tabular}
\caption{Numerical results when sample size $n=200$ and $n=400$ in sharp RD}
\end{table}

% latex table generated in R 3.3.1 by xtable 1.8-2 package
% Mon Oct 31 22:25:21 2016
\begin{table}[!ht]
\centering
\begin{tabular}{rr|r|rrr|rrrr}
  \hline
& & \multirow{2}{*}{Bias} & \multicolumn{3}{c|}{SE} & \multicolumn{4}{c}{Cover} \\
& & & NN & Plug-in & Boot & NN & Plug-in & Boot & Boot\_ED \\ 
  \hline
$n=250$ & IPCW & 0.051 & 1.271 & 1.147 & 1.269 & 0.906 & 0.876 & 0.896 & 0.942 \\ 
& Cox & 0.010 & 0.205 & 0.179 & 0.198 & 0.910 & 0.880 & 0.904 & 0.940 \\ 
& Log-norm & 0.012 & 0.215 & 0.188 & 0.207 & 0.918 & 0.880 & 0.882 & 0.936 \\ 
& Log-log & 0.011 & 0.215 & 0.188 & 0.207 & 0.920 & 0.880 & 0.890 & 0.942 \\ \hline
$n=500$ & IPCW & 0.136  & 0.916 & 0.831 & 1.363 & 0.926 & 0.913 & 0.950 & 0.952 \\ 
& Cox & 0.022 &  0.144 & 0.130 & 0.182 & 0.915 & 0.909 & 0.954 & 0.938 \\ 
& Log-norm & 0.026  & 0.150 & 0.136 & 0.193 & 0.934 & 0.924 & 0.954 & 0.948 \\ 
& Log-log & 0.026  & 0.150 & 0.136 & 0.194 & 0.932 & 0.920 & 0.954 & 0.946 \\ 
   \hline
\end{tabular}
\caption{Numerical results for mean response when sample size $n=250$ and $n=500$ in fuzzy RD}
\end{table}

\section{Real Data Analysis}
We now apply the proposed methodology to men who participated the Prostate, Lung, Colorectal and Ovarian (PLCO) cancer screening trial and randomized to receive annual prostate-specific antigen (PSA) screening for 6 years.  Among these 76,682 men received annual PSA screening from 1993 to 2001, those with a PSA of 4.0 ng/ml at anytime was recommended for further workup and biopsy, e.g., PSA-based screening strategy, for prostate cancer diagnosis.  In the context of RD design, this practice naturally creates a sharp RD design.   We therefore evaluated the role of additional workup and biopsy, as prompted by a PSA cutoff of 4.0 ng/ml, on patient survival and cancer incidence.  To simplify our discussions, here we focus on the role of PSA-based screening at the time of study entry among those who were tested for PSA previously.  While the role of PSA-based screening among those who had been exposed to PSA may be also of interest, its analysis involves methodology that is still being developed and will not be discussed here. \\
\indent Although the local randomization property holds for RD design, this property assumes that treatment assignment is independent of other covariates. To alleviate the potential impacts due to the associations between PSA-based screening strategy and covariates at study entry, we conducted propensity score matching with age, weight, height and BMI at the baseline.  Using matching, we have 2681 people in the two treatment groups, those with PSA level greater than 4.0 ng/ml and those with PSA level less than or equal to 4.0 ng/ml. In this sample, censoring rates for mortality and cancer incidence are 78.6\% and 61.3\%, respectively. \\
\indent We select DR approach with conditional expectation method on parametric AFT model. For parametric AFT model, we choose lognormal and loglogistic distributions. Table 3 shows results from sharp RD. We use nearest neighbor (NN) approach for calculation of standard error.
\begin{table}[!ht]
\centering
\begin{tabular}{ccccc} \hline
& \multicolumn{2}{c}{Mortality} & \multicolumn{2}{c}{First cancer incidence} \\
& Effect & SE & Effect & SE\\ \hline
Log-norm & 0.054 & 0.243 & -0.203 & 0.260 \\
Log-log & 0.043 & 0.195 & -0.193 & 0.242 \\ \hline
\end{tabular}
\caption{Results from PLCO dataset}
\label{tab_data}
\end{table}
Results show that there is no significant screening effect from baseline PSA threshold level 4.0mg/nL for mortality and cancer incidence. The screening effect with cutpoint 4.0mg/nL for mortality is slightly positive, which implies that screening slightly increase survival time although the effect is not statistically significant. On the other hand, screening decreases the time to first cancer incidence.  We also create a data-driven RD plot proposed by Calonico, Cattaneo and Titiunik (2015a) in Figure 1. This data-driven plot reflects the variability in the data by using local sample means with evenly-spaced bins. Although it is not intuitive to use our transformed response directly with theory in Calonico, Cattaneo and Titiunik (2015a), these plots are still useful to capture the variability of the transformed response with cutoffs and to check causal effect graphically. These two plots also indicate no treatment effect using a baseline PSA threshold level of 4.0mg/nL.  
\begin{figure}[!ht]
\centering
    \begin{minipage}{0.5\textwidth}
        \centering
        \includegraphics[scale = 0.35]{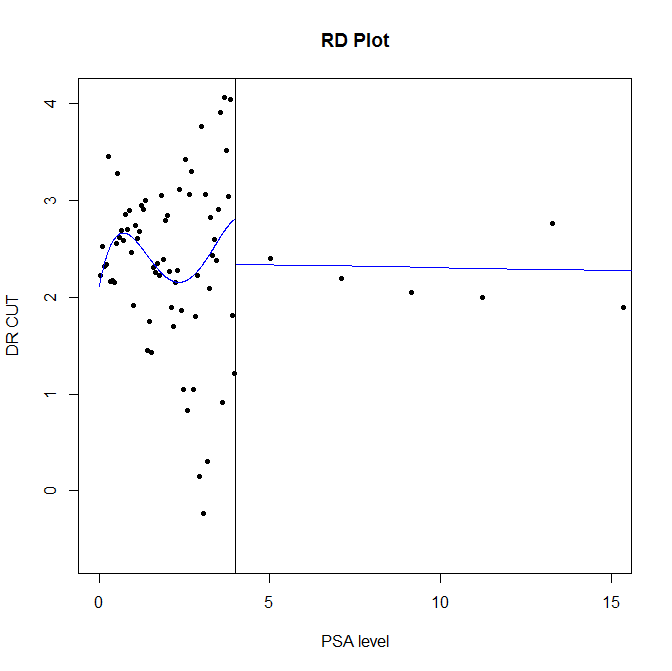}
    \end{minipage}%
    \begin{minipage}{0.5\textwidth}
        \centering
        \includegraphics[scale=0.35]{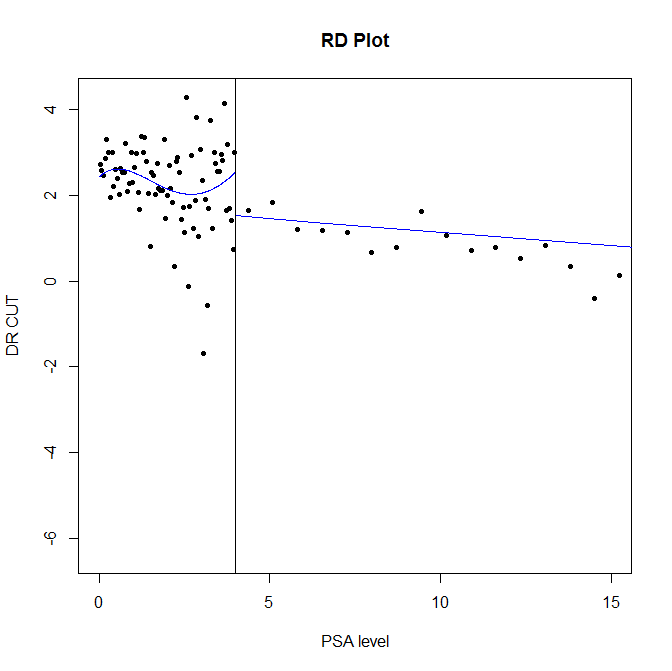}
    \end{minipage}
    \caption{Data-driven RD plot by Calonico, Cattaneo and Titiunik (2015a) for mortality (left) and the first cancer incidence (right)}
\end{figure}

\section{Conclusion} 
In this manuscript, we have proposed new estimators of causal effects with censored data in RD designs. Simulation studies show that the DR approach yields a more efficient estimator than IPCW. Moreover, the bias of DR method is smaller than one of DR method. Moreover, our method uses existing software, so researchers who easily apply this methodology. \\
\indent In this manuscript, only one forcing variable is considered for analysis. However, in practice data may contain several forcing variables, and they may provide additional information for treatment effect. There are two possible scenarios : (i) the forcing variable is a function of multiple covariates; (ii) there is one forcing variable correlated with other covariates as shown in our data analysis. Imbens and Zajonc (2009) and Zajonc (2012) study the situation of multiple forcing variables. Recently, Calonico et al. (2018) propose covariate adjustment approach in RD. It is great of interest to include covariates or consider the composite forcing variable in RD analysis. Our future work is to propose estimation procedure of treatment effect in RD adjusting for effects of other covariates. Although we only have one forcing variable without covariates, it is expected for the forcing variable to have multiple cutoffs. Cattaneo et al. (2016) discuss the multi-cutoff problem. This is also an interesting future work. \\ 
\indent We have adapted the LM approach for bandwidth selection. %However, as can be seen in Wang, Carroll and Lin (2005), undersmoothing may be required to obtain better asymptotic properties. 
Imbens and Kalyanaraman (2012) propose optimal bandwidth selection based on mean square error approximation and Calonico, Cattaneo and Titiunik (2014) propose bandwidth selection which helps bias correction. These have elegant asymptotic theory for their bandwidth selection proposals. This is currently under investigation.

%%%%%%%%%%%%%%%%%%%%%%%%%%%%%%%%%%%%%%%%%%%%%%%%%%%%%%%%%%%%%%%%%%%%%%%%%%%%%%%%%%%%%%%%%%%%%%%%%%%%%%%%%%%%%%%%%%%%%%%%%%%%

%%%%%%%%%%%%%%%%%%%%%%%%%%%%%%%%%%%%%%%%%%%%%%%%%%%%%%%%%%%%%%%%%%%%%%%%%%%%%%%%%%%%%%%%%%%%%%%%%%%%%%%%%%%%%%%%%%%%%%%%%%%%
\vskip 14pt
\noindent {\large\bf Acknowledgements}

\noindent The authors thank Dr. Jon Steingrimsson, Dr. Liqun Diao and Dr. Robert Strawderman for sharing computer code related to the IPCW and DR transformations.  
\par
\appendix

\section*{Supplementary Materials}
\section{Calculation of $\hat{G}$ and $\hat{Q}_Y(\cdot,\cdot)$}
In our estimation procedure, it is important to be able to model the censoring distribution and compute the relevant conditional expectation. We follow the approach from Steingrimsson et al.(2016). The censoring distribution can be modeled by a Kaplan-Meier estimator. The conditional expectation can be modeled in several different ways. Mathematically, for a general nondecreasing function $H$, it is represented as 
\begin{gather*}
E[H(T)|T \geq u,W] = \frac{\displaystyle \int_u^{\infty} H(T)dF(u|W)}{P(T \geq u|W)},
\end{gather*}
where $F$ is distribution of the failure time $T$. Various approaches can be taken to calculate this expectation.  One can use tree-based methods to find observations in the terminal node by the split, and then compute a Kaplan-Meier estimator based on the observations in the relevant terminal node. When using nonparametric and semiparametric methods, this integral is changed to summation at the time point when failure occurs. Note that in the nonparametric or semiparametric methods, the integral is computed as long as uncensored observation(s) greater than $u$ exists. If not, we let the estimator be maximum of observed time points (Steingrimsson et al. 2016). This kind of adjustment is required because the survival function estimated through a nonparametric or semiparametric method may not be proper, i.e., for a given time $t_0$, there may exist an interval with survival function being equal to $S(t_0)$. Estimation with parametric distributions is preferred because the parametric approach will always yield a proper survival function. Standard AFT models with parametric distributions (e.g. lognormal and log-logistic) can be used for the estimation (Steingrimsson et al. 2016). 

\section{Proof of asymptotic results}
\setcounter{equation}{0}

In this section, we provide proofs for the main asymptotic results in the paper. Since the proofs for IPCW and DR are similar, we mainly prove the result for DR. Moreover, results for IPCW depending on the censoring distribution only are similar. The proofs extend the arguments of Hahn, Todd and Van der Klaauw (1999). \\  
\indent Recall that $O = (\tilde{T},\Delta,W,Z)$. The observed data is given by $\mathcal{O} = \{O_i\}_{i=1}^n$ where $O_i = (\tilde{T}_i,\Delta_i,W_i,Z_i)$. Let $G$ and $S$ be survival probabilities of censoring and failure time based on possibly incorrect models. Let $G_0$ and $S_0$ be survival probabilities of censoring and failure time based on the true models for failure time and censoring, respectively. We assume that $G_0(u)$ and $S_0(u|w)$ are continuous and nonincreasing functions in $u$ for each $w$ with $0 \leq G_0(u), S_0(u|w) \leq 1$. Furthermore, we assume that $G(u)$ and $S(u|w)$ are right-continuous and nonincreasing in $u$ for each $w$ with $0 \leq G(u), S(u|w) \leq 1$ and $G_0(0) = G(0) = S_0(0|w) = S(0|w) = 1$. Define $F_0(u|w) = 1-S_0(u|w)$, $\bar{G}_0(u) = 1-G_0(u)$, $F(u|w) = 1-S(u|w)$ and $\bar{G}(u) = 1-G(u)$. Now we express $Q_Y(\cdot,\cdot)$ by $Q_Y(\cdot,\cdot,S)$ to indicate dependence of conditional expectation on failure time distribution. We first assume the following regularity conditions, similar to those in Steingrimsson, Diao, and Strawderman (2019):
\begin{itemize}
%\item Conditions 1-4 and 7 in Hahn et al. (2001).
\item[(C1)] $I_1 = \displaystyle \int_0^{\infty} \log(u) \frac{G_0(u)}{G(u-)}dF_0(u|w) < \infty$
\item[(C2)] For $a > 0$, 
\begin{gather*}
D_1(a) = \displaystyle \int_0^{a} \frac{S_0(u|w)}{S(u|w)}\frac{d\bar{G}_0(u)}{G(u-)} < \infty \quad D_2(a) = \displaystyle \int_0^{a} \frac{G_0(u)S_0(u|w)}{G(u)S(u|w)}\frac{d\bar{G}(u)}{G(u-)} < \infty
\end{gather*}
\item[(C3)] $I_2 = \displaystyle \int_0^{\infty} \log(u)[D_1(a-) - D_2(a-)]dF(u|w) < \infty$
\item[(C4)] $\displaystyle \int_0^{\infty} \frac{[\log(u)]^2}{G_0(u)}dF_0(u|w) < \infty$
\item[(C5)] $D_3(a) = \displaystyle \int_0^a \frac{Q_Y(u,w,S)}{\{G_0(u)\}^2}d\bar{G}_0(u) < \infty$ for each $a > 0$.
\item[(C6)] $\hat{G}$ is uniformly consistent to $G_0$. %i.e., $\sup_u |\hat{G}(u) - G_0(u)| \overset{p}{\rightarrow} 0$.
\item[(C7)] $\hat{S}$ is uniformly consistent to $S^*$ where $S^*$ is possibly incorrect model of $S$. %i.e., $\sup_u |\hat{S}(u|w) - S^*(u|w)| \overset{p}{\rightarrow} 0$.
%\item [(Modified condition 5 of Hahn, Todd and Klaauw (2001))] Define
\end{itemize}
As Steingrimsson, Diao, and Strawderman (2019) have shown, from conditions (C1)-(C5), we can prove that 
\begin{gather*}
E(Y_{DR}(O;G_0,S)|W) = E(Y_{DR}(O;G,S_0)|W) = E(Y_{DR}(O;G_0,S_0)|W) = E(Y|W) = \mu(W).
\end{gather*} 
This is necessary for proving asymptotic normality of $\hat{\tau}_{FRD}^{IPCW}(\hat{G})$, $\hat{\tau}_{FRD}^{DR}(\hat{G},\hat{S})$, $\hat{\tau}_{SRD}^{IPCW}(\hat{G})$, $\hat{\tau}_{SRD}^{DR}(\hat{G},\hat{S})$. Let $p(w) = E(Z|W=w)$ and define 
\begin{align*}
& \mu^+(w) = \lim_{w \downarrow w_0} E(Y|W=w) \quad \mu^-(w) = \lim_{w \uparrow w_0} E(Y|W=w) \\%= \lim_{w \downarrow w_0} E(Y_{IPCW}(O;G_0)|W=w) = \lim_{w \downarrow w_0} E(Y_{DR}(O;G_0,S)|W=w) = \lim_{w \downarrow w_0} E(Y_{DR}(O;G,S_0)|W=w) = \lim_{w \downarrow w_0} E(Y_{DR}(O;G_0,S_0)|W=w)\\
%& \mu^-(w) = \lim_{w \uparrow w_0} E(Y|W=w) %= \lim_{w \uparrow w_0} E(Y_{IPCW}(O;G_0)|W=w) = \lim_{w \uparrow w_0} E(Y_{DR}(O;G_0,S)|W=w) = \lim_{w \uparrow w_0} E(Y_{DR}(O;G,S_0)|W=w) = \lim_{w \uparrow w_0} E(Y_{DR}(O;G_0,S_0)|W=w) \\
& p^+(w) = \lim_{w \downarrow w_0} P(Z=1|W=w) \quad p^-(w) = \lim_{w \uparrow w_0} P(Z=1|W=w).
\end{align*}
Now we need conditions similar to Hahn, Todd and Van der Klaauw (1999). Define
\begin{gather*}
\hspace*{-20mm}
Y_{IPCW}(O;G) = \frac{\Delta Y}{G(T)} \quad Y^{DR}(O;G,S) = \frac{\Delta Y}{G(T)} + \int_0^{\tilde{T}} \frac{Q_Y(u,W,S)}{G(u)}dM_G(u) \\
Y_{IPCW^*}(O;G) = Y_{IPCW}(O;G) - \mu^+(w_0) - \mu'^+(w_0)(W - w_0) \\
Y_{DR^*}(O;G,S) = Y^{DR}(O;G,S) - \mu^+(w_0) - \mu'^+(w_0)(W - w_0) \\
Z^* = Z - p^+(w_0) - p'^+(w_0)(W - w_0) \\
L_{ih}^+ = I(W_i \geq w_0) K\bigg{(}\frac{W_i- w_0}{h}\bigg{)}  \quad L_{ih}^- = I(W_i < w_0) K\bigg{(}\frac{W_i- w_0}{h}\bigg{)}.
\end{gather*}
We can define empirical quantities of $Y_{IPCW}(O;G),Y_{IPCW^*}(O;G),Y_{DR}(O;G,S),Y_{DR^*}(O;G,S),Z^*$. Denote the empirical quantities by $Y_{IPCW,i}(O_i;G),Y_{IPCW^*,i}(O_i;G),Y_{DR,i}(O_i;G,S),Y_{DR^*,i}(O_i;G,S),Z_i^*$. We further define
\begin{gather*}
\sigma_{DR}^2(w;G,S) = Var(Y_{DR}(O;G,S)|W=w) \quad \sigma_{DR}^{2+}(w_0;G,S) = \lim_{\epsilon \downarrow w_0} Var(Y_{DR}(O;G,S)|W=w) \\
\sigma_{DR}^{2-}(w_0;G,S) = \lim_{\epsilon \uparrow w_0} Var(Y_{DR}(O;G,S)|W=w) \quad \eta_{DR}(w;G,S) = Cov(Y_{DR}(O;G,S), Z|W=w) \\
\hspace*{-15mm}
\eta_{DR}^+(w_0;G,S) = \lim_{w \downarrow w_0} Cov(Y_{DR}(O;G,S), Z|W=w) \quad \eta_{DR}^-(w_0;G,S) = \lim_{w \uparrow w_0} Cov(Y_{DR}(O;G,S), Z|W=w).
\end{gather*}
We can make similar definitions for $\sigma_{IPCW}^2(G,S),\sigma_{IPCW}^{2+}(w_0;G,S),\sigma_{IPCW}^{2-}(w_0;G,S),\eta_{IPCW}(w_0;G,S),\\
\eta_{IPCW}^+(w_0;G,S)$ and $\eta_{IPCW}^-(w_0;G,S)$. Define matrices
\begin{gather*}
\mathbf{X}_h = 
\begin{pmatrix}
1 & \dfrac{W_1 - w_0}{h} \\
1 & \dfrac{W_2 - w_0}{h} \\
\vdots & \vdots \\
1 & \dfrac{W_n - w_0}{h} 
\end{pmatrix}
\\
\hspace*{-15mm}
\mathbf{W}_h^+ = 
\begin{pmatrix}
I(W_1 \geq w_0) K\bigg{(}\dfrac{W_1- w_0}{h}\bigg{)} & 0 & 0 & \ldots & 0 \\
0 & I(W_2 \geq w_0) K\bigg{(}\dfrac{W_i- w_0}{h}\bigg{)} & 0 & \ldots & 0 \\
\vdots & \vdots & \vdots & \ddots & 0 \\
0 & 0 & \ldots & 0 & I(W_n \geq w_0) K\bigg{(}\dfrac{W_i- w_0}{h}\bigg{)} 
\end{pmatrix}.
%\\
%\mathbf{W}_h^- = 
%\begin{pmatrix}
%I(W_1 < w_0) K\bigg{(}\dfrac{W_1- w_0}{h}\bigg{)} & 0 & 0 & \ldots & 0 \\
%0 & I(W_2 < w_0) K\bigg{(}\dfrac{W_i- w_0}{h}\bigg{)} & 0 & \ldots & 0 \\
%\vdots & \vdots & \vdots & \ddots & 0 \\
%0 & 0 & \ldots & 0 & I(W_n < w_0) K\bigg{(}\dfrac{W_i- w_0}{h}\bigg{)} \\
%\end{pmatrix}
\end{gather*}

\begin{itemize}
\item[(R1)] For $W \ne w_0$, let $\mu(z)$ and $p(z)$ be twice continuously differentiable functions. Let $\mu'(w)$ and $\mu''(w)$ be the first and second derivatives of $\mu(w)$ and similarly for $p'(w)$ and $p''(w)$. Let $\mu'^+(w)$ and $\mu''^+(w)$ be the first and second derivatives of $\mu^+(w)$, and $p'^+(w)$ and $p''^+(w)$ are the first and second derivatives of $p(w)$.  Define $\mu'^-(w)$ and $\mu''^-(w)$ to be the first and second derivative of $\mu^-(w)$ and $p'^-(w)$ and $p''^-(w)$ are first and second derivative of $p^-(w)$. Assume there exists $B >0$ such that $|\mu^+(w)|,|\mu'^+(w)|,|\mu''^+(w)|$ and $|p^+(w)|,|p'^+(w)|,|p''^+(w)|$ are uniformly bounded on $(w_0,w_0+B]$. Similarly, $|\mu^-(w)|,|\mu'^-(w)|,|\mu''^-(w)|$ and $|p^+(w)|,|p'^+(w)|,|p''^+(w)|$ are uniformly bounded on $[w_0-B,w_0)$. 
\item[(R2)] Assume that $\mu^+(w_0), \mu'^+(w_0), \mu''^+(w_0), \mu^-(w_0). \mu'^-(w_0), \mu''^-(w_0)$, $p^+(w_0), p'^+(w_0), p''^+(w_0), \\
p^-(w_0), p'^-(w_0)$ and $p''^-(w_0)$ are finite.
\item[(R3)] Let $g(w)$ be the common density of $W_i$. Assume that $g(w)$ is continuous and bounded away from zero in a neighborhood of $w_0$.
\item[(R4)] $\sigma_{IPCW}^2(w;G_0),\sigma_{DR}^2(w;G_0,S^*)$ and $\eta_{IPCW}(w;G_0), \eta_{DR}(w;G_0,S^*)$ are uniformly bounded in a neighborhood of $w_0$. 
\item[(R5)] Assume that $\sigma_{IPCW}^{2+}(w_0;G_0), \sigma_{DR}^{2+}(w_0;G_0, S^*), \sigma_{IPCW}^{2-}(w_0;G_0), \sigma_{DR}^{2-}(w_0;G_0,S^*)$ and $\eta_{IPCW}^{+}(w_0;G_0),\eta_{DR}^+(w_0;G_0,S^*)$, $\eta_{IPCW}^{-}(w_0;G_0),\eta_{DR}^-(w_0;G_0,S^*)$ are finite.
\item[(R6)] $\lim\limits_{W_i \uparrow w_0} E\bigg{[}|Y_{IPCW,i}(O_i;G_0) - \mu(W_i)|^r \bigg{|}W_i\bigg{]}$ and $\lim\limits_{W_i \uparrow w_0} E\bigg{[}|Y_{DR,i}(O_i;G_0,S^*) - \mu(W_i)|^r \bigg{|}W_i\bigg{]}, r=1,2,3$ are finite. We assume similarly when $W_i \downarrow w_0$.
\item[(R7)] $K$ is continuous and symmetric. Moreover, support of $K$ is compact and for any $u$, $K(u) \geq 0$. 
\item[(R8)] The bandwidth satisfies $h \sim n^{-1/5}$ where $\sim$ indicates ``asymptotically equivalent". 
\item[(R9)] Let $V_n = o_p(1)$. Then
\begin{gather*}
E\bigg{[}\bigg{(}\frac{W_i - w_0}{h} \bigg{)}^{j_1} (L_{ih}^{+})^{j_2}V_n\bigg{]} = O(1) \quad j_1=0,\ldots,6 \quad j_2 = 1,2,3.
\end{gather*} 
\end{itemize}
Conditions (R1) - (R8) are standard ones for proving asymptotic normality of sharp and fuzzy RD estimators without censoring. (R9) is a required assumption for making $o_p(1)$ term from consistency of $\hat{G}$ and $\hat{S}$.\\
Modified versions of Lemma 1- Lemma 6 in Hahn, Todd and Van der Klaauw (1999) are required. Let $W_{ih}^r = \bigg{(}\dfrac{W_i - w_0}{h_n}\bigg{)}^r,r=0,1$.
Consider 
%\begin{gather*}
%\mathbf{Y}_{IPCW^*,h}(G) = \mathbf{X}_h^T\mathbf{W}_h\mathbf{Y}_{IPCW^*}(G) \\
%\mathbf{Y}_{DR^*,h}(G) = \mathbf{X}_h^T\mathbf{W}_h\mathbf{Y}_{DR^*}(G)
%\mathbf{Z}_{i,h}^* = \mathbf{X}_h^T\mathbf{W}_h\mathbf{Z}
%\end{gather*}
\begin{gather*}
\mathbf{A}_{IPCW^*,i,h}^+(G) = \begin{pmatrix} W_{ih}^0 Y_{IPCW^*,i}(O_i;G) L_{ih}^+ \\ W_{ih}^1 Y_{IPCW^*,i}(O_i;G) L_{ih}^+ \end{pmatrix} \quad 
\mathbf{A}_{DR^*,i,h}^+(G,S) = \begin{pmatrix} W_{ih}^0 Y_{DR^*,i}(O_i;G,S) L_{ih}^+ \\ W_{ih}^1 Y_{DR^*,i}(O_i;G,S) L_{ih}^+)^T \end{pmatrix} \\
\mathbf{A}_{Z^*,i,h}^+ = \begin{pmatrix} W_{ih}^0 Z_i^* L_{ih}^+ \\ W_{ih}^1 Z_i^* L_{ih}^+ \end{pmatrix}. %\quad \mathbf{A}_{Z^*,i,h}^- = \begin{pmatrix} W_{ih}^0 Z_i^* L_{ih}^- \\ W_{ih}^1 Z_i^* L_{ih}^- \end{pmatrix}
\end{gather*}
Since the proof will be similar for $W < w_0$, we will only provide proofs for the case $W \geq w_0$. In the following lemmas, we assume that conditions (C1) - (C5) and (R1) - (R8) hold. 
Lemma 1 in Hahn, Todd and Van der Klaauw (1999) does not involve any modeling of failure time and censoring, so it clearly holds in our case.
\begin{lemma}
\begin{gather*}
\frac{1}{nh}(\mathbf{X}_h^T\mathbf{W}_h^+\mathbf{X}_h) \rightarrow g(w_0)\boldsymbol{\Omega}.
\end{gather*}
where $\boldsymbol{\Omega} = \begin{pmatrix} \gamma_0 & \gamma_1 \\ \gamma_1 & \gamma_2 \end{pmatrix}$ and $\gamma_j = \displaystyle \int_0^{\infty} u^j K(u)du, j=0,1,2$.
\end{lemma}
\begin{proof}
See Hahn, Todd and Van der Klaauw (1999).
\end{proof}
\begin{lemma}
Two terms
\begin{gather*}
E\bigg{\{}\frac{1}{nh}\sum_{i=1}^n \mathbf{A}_{IPCW^*,i,h}^+(O_i;G_0)\bigg{\}} \quad E\bigg{\{}\frac{1}{nh}\sum_{i=1}^n \mathbf{A}_{DR^*,i,h}^+(O_i;G_0,S_0)\bigg{\}} 
%E\bigg{\{}\frac{1}{nh}\sum_{i=1}^n \mathbf{A}_{DR^*,i,h}^+(O_i;G,S_0)\bigg{\}}, E\bigg{\{}\frac{1}{nh}\sum_{i=1}^n \mathbf{A}_{DR^*,i,h}^+(O_i;G_0,S_0)\bigg{\}}
\end{gather*}
converge to  $\frac{1}{2}g(z_0)\mu''^{+}(w_0)h^2(\boldsymbol{\delta} + o(1))$ and $E\bigg{\{}\dfrac{1}{nh} \displaystyle \sum_{i=1}^n \mathbf{A}_{Z^*,i,h}^+\bigg{\}}$ converges to $\frac{1}{2}g(z_0)p''^{+}(w_0)h^2(\boldsymbol{\delta} + o(1))$ 
%\begin{gather*}
%E\bigg{\{}\frac{1}{nh}\sum_{i=1}^n \mathbf{A}_{IPCW^*,i,h}^-(O_i;G_0)\bigg{\}}, %E\bigg{\{}\frac{1}{nh}\sum_{i=1}^n \mathbf{A}_{DR^*,i,h}^-(O_i;G_0,S)\bigg{\}}, \\
%E\bigg{\{}\frac{1}{nh}\sum_{i=1}^n \mathbf{A}_{DR^*,i,h}^-(O_i;G,S_0)\bigg{\}}, E\bigg{\{}\frac{1}{nh}\sum_{i=1}^n \mathbf{A}_{DR^*,i,h}^-(O_i;G_0,S_0)\bigg{\}}
%\end{gather*}
%converge to  $\frac{1}{2}g(z_0)\mu''^{-}(w_0)h^2(\boldsymbol{\delta} + o(1))$ and $E\bigg{\{}\dfrac{1}{nh} \displaystyle \sum_{i=1}^n Z_{i,h}^*\bigg{\}}$ converges to $\frac{1}{2}g(z_0)\mu''^{-}(w_0)h^2(\boldsymbol{\delta} + o(1))$
where
\begin{gather*}
\boldsymbol{\delta} = \begin{pmatrix} \delta_0 \\ \delta_1 \end{pmatrix} = \begin{pmatrix} \int_0^{\infty} u^2 K(u)du \\ \int_0^{\infty} u^3 K(u)du \end{pmatrix}.
\end{gather*}
\end{lemma}  
\begin{proof}
We look at individual terms in $\mathbf{A}_{IPCW^*,i,h}^+(\mathcal{O};G_0)$ and $\mathbf{A}_{DR^*,i,h}^+(\mathcal{O};G_0,S)$. Let $q=0,1,2$ and for any model for $S$,
\begin{gather*}
U_{IPCW}^{q+}(\mathcal{O};G_0) = \frac{1}{nh} \sum_{i=1}^n \bigg{(}\frac{W_i-w_0}{h}\bigg{)}^qY_{IPCW^*,i}(O_i;G_0) L_{ih}^+ \\ 
U_{DR}^{q+}(\mathcal{O};G_0,S) = \frac{1}{nh}\sum_{i=1}^n \bigg{(}\frac{W_i-w_0}{h}\bigg{)}^qY_{DR^*,i}(O_i;G_0,S) L_{ih}^+.
\end{gather*}
Since $E(Y_{IPCW,i}(G_0)|W_i) = E(Y_{DR,i}(G_0,S)|W_i) = E(Y_{DR,i}(G,S_0)|W_i) = E(Y_{DR,i}(G_0,S_0)|W_i) = \mu(W_i)$,
\begin{gather*}
\{U_{IPCW}^{q+}(\mathcal{O};G_0)\} = E\{U_{DR}^q(\mathcal{O};G_0,S)\}  \\
= \frac{1}{h}E\bigg{[}\bigg{(}\frac{W_i - w_0}{h}\bigg{)}^q\bigg{(}\frac{1}{2}\mu''^{+}(w_0) \cdot (W_i - w_0)^2 + \zeta^Y(W_i)\bigg{)}L_{ih}^+\bigg{]}, 
\end{gather*}
where 
\begin{gather*}
\zeta^Y(w) = \mu(w) - \mu^+(w_0) - \mu'^{+}(w_0)(w - w_0) -\frac{1}{2}\mu''^{+}(w_0) \cdot (w - w_0)^2.
\end{gather*}
By condition (C6) and (C7), 
\begin{gather*}
E(U_{DR}^{q+}(\mathcal{O};\hat{G},\hat{S})) = \frac{1}{h}E\bigg{[}\bigg{(}\frac{W_i - w_0}{h}\bigg{)}^q\bigg{(}\frac{1}{2}\mu''^{+}(w_0) \cdot (W_i - w_0)^2 + \zeta^Y(W_i) + o_p(1) + o_p(1) \bigg{)}L_{ih}^+\bigg{]}
\end{gather*}

By Taylor series expansion, $\sup\limits_{w_0 < w \leq w_0+Bh} |\zeta^Y(w)| = o(h^2)$. Hence by the arguments of Hahn, Todd and Van der Klaauw (1999), 
\begin{gather*}
\frac{1}{h}E\bigg{[}\bigg{(}\frac{W_i - w_0}{h}\bigg{)}^q\bigg{(}\frac{1}{2}\mu''^{+}(w_0) (W_i - w_0)^2+ o_p(1)\bigg{)}L_{ih}^+ \bigg{]} = \\
= \frac{1}{2h}\mu''^{+}(w_0)  \int_{w_0}^{\infty} \bigg{(}\frac{W_i - w_0}{h}\bigg{)}^q (W_i - w_0)^2 K\bigg{(}\frac{W_i-w_0}{h}\bigg{)}g(w)dw + o(h^2) \\ = \delta_qg(w_0) h^2 + o(h^2).
\end{gather*} 
We can apply the same arguments for $\mathbf{A}_{DR^*,i,h}(\mathcal{O};\hat{G},\hat{S})$

Define $U_Z^{q} = \dfrac{1}{nh} \displaystyle \sum_{i=1}^n \bigg{(}\dfrac{W_i-w_0}{h}\bigg{)}^qZ_{i}^* L_{ih}^+$ and $\zeta^Z(w) = p(w) - p^+(w_0) - p'^{+}(w_0)(w - w_0) -\frac{1}{2}p''^{+}(w_0) \cdot (w - w_0)^2$. Clearly, $\sup\limits_{w_0 < w \leq w_0+Bh} |\zeta^Z(w)| = o(h^2)$. Hence
\begin{align*}
& E(U_Z^{q}) = E \bigg{(} \frac{1}{nh} \displaystyle \sum_{i=1}^n \bigg{(}\frac{W_i-w_0}{h}\bigg{)}^qZ_{i}^* L_{ih}^+ \bigg{)} \\
& = \frac{1}{h}E\bigg{[}\bigg{(}\frac{W_i - w_0}{h}\bigg{)}^q\bigg{(}\frac{1}{2}p''^{+}(w_0) (W_i - w_0)^2+\zeta^Z(W_i)\bigg{)}L_{ih}^+\bigg{]} \\
& = \frac{1}{2h}p''^{+}(w_0)  \int_{w_0}^{\infty} \bigg{(}\frac{W_i - w_0}{h}\bigg{)}^k (W_i - w_0)^2 K\bigg{(}\frac{W_i-w_0}{h}\bigg{)}g(w)dw + o(h^2) = \delta_qg(w_0) h^2 + o(h^2).
\end{align*}
\end{proof}  

\begin{lemma}
Let 
\begin{align*}
\mathbf{A}_{IPCW^*,h}^+(W_i;G) & = E(\mathbf{A}_{IPCW^*,i,h}^+(G)|W_i) \quad \mathbf{A}_{DR^*,h}^+(W_i;G,S) = E(\mathbf{A}_{DR^*,i,h}^+(G,S)|W_i) \\
 \mathbf{A}_{Z^*,h}^+(W_i) & = E(\mathbf{A}_{Z^*,i,h}^+|W_i) \quad \mathbf{A}_{IPCW^*,h}^-(W_i;G)  = E(\mathbf{A}_{IPCW^*,i,h}^-(G)|W_i) \\
 \mathbf{A}_{DR^*,h}^-(W_i;G,S) &= E(\mathbf{A}_{DR^*,i,h}^-(G,S)|W_i) \quad \mathbf{A}_{Z^*,h}^-(W_i) = E(\mathbf{A}_{Z^*,i,h}^-|W_i) 
\end{align*}
Then
\begin{gather*}
\frac{1}{nh} \sum_{i=1}^n \mathbf{A}_{IPCW^*,h}^+(W_i;\hat{G})  = E\bigg{\{}\frac{1}{nh} \sum_{i=1}^n \mathbf{A}_{IPCW^*,i,h}^+(G_0)\bigg{\}} + o_p(h^2) \\
\frac{1}{nh} \sum_{i=1}^n \mathbf{A}_{DR^*,h}^+(W_i;\hat{G},\hat{S}) = E\bigg{\{}\frac{1}{nh} \sum_{i=1}^n \mathbf{A}_{DR^*,i,h}^+(G_0,S^*)\bigg{\}} + o_p(h^2) \\
\frac{1}{nh} \sum_{i=1}^n \mathbf{A}_{Z^*,h}^+(W_i) =E\bigg{(}\frac{1}{nh} \sum_{i=1}^n \mathbf{A}_{Z^*,i,h}^+ \bigg{)} + o_p(h^2).
\end{gather*}

\end{lemma}
\begin{proof}

%Since 
%\begin{gather*}
%\frac{1}{nh} \sum_{i=1}^n \mathbf{A}_{IPCW^*,h}^+(W_i;G_0) = \frac{1}{nh} \sum_{i=1}^n \mathbf{X}_{i,h} L_{ih}^+ \bigg{(}\frac{1}{2}\mu''^{+}(w_0)(W_i - w_0)^2 + \zeta^Y(W_i) \bigg{)}
%\end{gather*}
\begin{gather*}
\frac{1}{nh} \sum_{i=1}^n \mathbf{A}_{DR^*,h}^+(W_i;\hat{G},\hat{S}) =  \frac{1}{nh} \sum_{i=1}^n \boldsymbol{X}_{i,h}L_{i,h}^+\bigg{(}\frac{1}{2}\mu''^{+}(w_0)(W_i - w_0)^2 + \zeta^Y(W_i) + E_n^*\bigg{)}
\end{gather*}
where $E_n^* = o_p(1)$ and $\mathbf{X}_{i,h} = \begin{pmatrix} 1 \\ \frac{W_i-w_0}{h} \end{pmatrix}$.
Then using similar calculations to those in Hahn, Todd and Van der Klaauw (1999), 
\begin{gather*}
Var\bigg{\{}\frac{1}{nh} \sum_{i=1}^n \bigg{(}\frac{W_i-w_0}{h}\bigg{)}^q L_{ih}^+ \bigg{(}\frac{1}{2}\mu''^{+}(w_0)(W_i - w_0)^2 + \zeta^Y(W_i)  + E_n^* \bigg{)}\bigg{\}} \\
= \frac{1}{nh} \bigg{[} E\bigg{\{} \bigg{(}\frac{W_i-w_0}{h}\bigg{)}^q L_{ih}^+ \bigg{(}\frac{1}{2}\mu''^{+}(w_0)(W_i - w_0)^2 + \zeta^Y(W_i) + E_n^*\bigg{)}\bigg{\}}^2 \\
- \bigg{[}E\bigg{\{}\bigg{(}\frac{W_i-w_0}{h}\bigg{)}^q L_{ih}^+ \bigg{(}\frac{1}{2}\mu''^{+}(w_0)(W_i - w_0)^2 + \zeta^Y(W_i) +  E_n^*\bigg{)}\bigg{\}}\bigg{]}^2 \bigg{]}.
\end{gather*}
Then we can choose constant $\varrho$ such that
\begin{gather*}
\frac{1}{nh} \bigg{[} E\bigg{\{} \bigg{(}\frac{W_i-w_0}{h}\bigg{)}^q L_{ih}^+ \bigg{(}\frac{1}{2}\mu''^{+}(w_0)(W_i - w_0)^2 + \zeta^Y(W_i) + E_n^* \bigg{)}\bigg{\}}^2 \\
- \bigg{[}E\bigg{\{}\bigg{(}\frac{W_i-w_0}{h}\bigg{)}^q L_{ih}^+ \bigg{(}\frac{1}{2}\mu''^{+}(w_0)(W_i - w_0)^2 + \zeta^Y(W_i) + E_n^* \bigg{)}\bigg{\}}\bigg{]}^2 \bigg{]} \\
\leq \frac{1}{nh} \varrho E\bigg{\{} \bigg{(}\frac{W_i-w_0}{h}\bigg{)}^{2q} (L_{ih}^+)^2 \bigg{(}(W_i - w_0)^4 + \{\zeta^Y(W_i)\}^2 \bigg{)}\bigg{\}}^2.
\end{gather*}
From Hahn, Todd and Van der Klaauw (1999), we have that 
\begin{gather*}
\frac{1}{nh} \varrho E\bigg{\{} \bigg{(}\frac{W_i-w_0}{h}\bigg{)}^{2q} (L_{ih}^+)^2 \bigg{(}(W_i - w_0)^4 + \{\zeta^Y(W_i)\}^2 \bigg{)}\bigg{\}}^2 = o(h^2).
\end{gather*}
Hence
\begin{gather*}
\frac{1}{nh} \sum_{i=1}^n \mathbf{A}_{DR^*,h}^+(W_i;\hat{G},\hat{S})  = E\bigg{\{}\frac{1}{nh} \sum_{i=1}^n \mathbf{A}_{DR^*,i,h}^+(G_0,S^*)\bigg{\}} + o_p(h^2).
\end{gather*}
Similar arguments can be applied to $\dfrac{1}{nh} \displaystyle \sum_{i=1}^n \mathbf{A}_{IPCW^*,h}^+(W_i;\hat{G})$ and $\dfrac{1}{nh} \displaystyle \sum_{i=1}^n \mathbf{A}_{Z^*,h}^+(W_i)$.
\end{proof}
\begin{lemma}
Let $v_q = \int_0^{\infty} u^q\{K(u)\}^2, q=0,1,2$ and
\begin{gather*}
\mathbf{\bar{A}}_{IPCW^*,h}^+(\mathcal{O};\hat{G}) =  \frac{1}{nh} \sum_{i=1}^n \{\mathbf{A}_{IPCW^*,i,h}^+(O_i;\hat{G}) - \mathbf{A}_{IPCW^*,h}^+(W_i;\hat{G})\} \\
\mathbf{\bar{A}}_{DR^*,h}^+(\mathcal{O};\hat{G},\hat{S}) = \frac{1}{nh} \sum_{i=1}^n \{\mathbf{A}_{DR^*,i,h}^+(O_i;\hat{G},\hat{S}) - \mathbf{A}_{DR^*,h}^+(W_i;\hat{G},\hat{S})\} \\
\mathbf{\bar{A}}_{Z^*,h}^+ = \frac{1}{nh} \sum_{i=1}^n \{\mathbf{A}_{Z^*,i,h}^+ - \mathbf{A}_{Z^*,h}^+(W_i)\}. 
\end{gather*}
Then
\begin{gather*}
\hspace*{-15mm}
Var\{\mathbf{\bar{A}}_{IPCW^*,h}^+(\mathcal{O};G_0)\} = \frac{1}{nh}\sigma_{IPCW}^{2+}(w_0;G_0)g(w_0)\mathcal{V} \quad Var\{\mathbf{\bar{A}}_{DR^*,h}^+(\mathcal{O};G_0,S)\} = \frac{1}{nh}\sigma_{DR}^{2+}(w_0;G_0,S)g(w_0)\mathcal{V} \\
Var(\mathbf{\bar{A}}_{Z^*,h}^+) = \frac{1}{nh}p^+(w_0)\{1-p^+(w_0)\}g(w_0)\mathcal{V} \\
Cov(\mathbf{\bar{A}}_{IPCW^*,h}^+(\mathcal{O};G_0), \mathbf{\bar{A}}_{Z^*,h}^+) = \frac{1}{nh}\eta_{IPCW}^{+}(w_0;G_0)g(w_0)\mathcal{V} \\
 Cov(\mathbf{\bar{A}}_{DR^*,h}^+(\mathcal{O};G_0,S), \mathbf{\bar{A}}_{Z^*,h}^+) = \frac{1}{nh}\eta_{DR}^{+}(w_0;G_0,S)g(w_0)\mathcal{V}, 
\end{gather*}
where 
\[
\mathcal{V} =
\begin{pmatrix}
v_0 + o(1) & v_1 + o(1) \\
v_1 + o(1) & v_2 + o(1)
\end{pmatrix}. 
\]
\end{lemma}
\begin{proof}
Let $q=0,1,2$ and 
\begin{gather*}
\mathbf{A}_{i,IPCW}^{ctd,+}(O_i;G) = \begin{pmatrix} W_{ih}^0 L_{ih}^+ (Y_{IPCW,i}(O_i;G) - \mu(W_i)) \\ W_{ih}^1 L_{ih}^+ (Y_{IPCW,i}(O_i;G) - \mu(W_i)) \end{pmatrix} \\
\mathbf{A}_{i,DR}^{ctd,+}(O_i;G,S) = \begin{pmatrix} W_{ih}^0 L_{ih}^+ (Y_{DR,i}(O_i;G,S) - \mu(W_i)) \\ W_{ih}^1 L_{ih}^+ (Y_{DR,i}(O_i;G,S) - \mu(W_i)) \end{pmatrix}. 
\end{gather*}
By expansions, 
\begin{gather}
\mathbf{\bar{A}}_{DR^*,h}^+(\mathcal{O};\hat{G},\hat{S}) = \frac{1}{nh} \sum_{i=1}^n \{\mathbf{A}_{DR^*,i,h}^+(O_i;\hat{G},\hat{S}) - \mathbf{A}_{DR^*,h}^+(W_i;\hat{G},\hat{S})\} \nonumber \\
= \frac{1}{nh} \sum_{i=1}^n \{\mathbf{A}_{DR^*,i,h}^+(O_i;\hat{G},\hat{S}) - \mathbf{A}_{DR^*,h}^+(W_i;G_0,S^*) + \mathbf{A}_{DR^*,h}^+(W_i;G_0,S^*) - \mathbf{A}_{DR^*,h}^+(W_i;\hat{G},\hat{S})\} \nonumber \\
%= \frac{1}{nh} \sum_{i=1}^n \{\mathbf{A}_{DR^*,i,h}^+(O_i;\hat{G},\hat{S}) - \mu(W_i) + \mu(W_i) - \mathbf{A}_{DR^*,h}^+(W_i;\hat{G},\hat{S})\} 
= \frac{1}{nh} \sum_{i=1}^n \mathbf{X}_{i,h}L_{ih}^+\{Y_{DR^*,i}(O_i;\hat{G},\hat{S}) - E(Y_{DR^*,i}(O_i;G_0,S^*)|W_i) \nonumber \\ + E(Y_{DR^*,i}(O_i;G_0,S^*)|W_i) - E(Y_{DR^*,i}(O_i;\hat{G},\hat{S})|W_i)\} \nonumber \\
= \frac{1}{nh} \sum_{i=1}^n \mathbf{X}_{i,h}L_{ih}^+\{Y_{DR^*,i}(O_i;\hat{G},\hat{S}) - E(Y_{DR^*,i}(O_i;G_0,S^*)|W_i) \nonumber \\ + E(Y_{DR^*,i}(O_i;G_0,S^*)|W_i) - E(Y_{DR^*,i}(O_i;\hat{G},\hat{S})|W_i)\}
\label{eq:eqlem4}
\end{gather}
By consistency of $\hat{G}$ and $\hat{S}$, 
\begin{gather*}
%\eqref{eq:eqlem4} = \frac{1}{nh} \sum_{i=1}^n \{\mathbf{A}_{DR^*,i,h}^+(O_i;\hat{G},\hat{S}) - \mu(W_i) +O(1)\} = \frac{1}{nh} \sum_{i=1}^n \{\mathbf{A}_{i,DR}^{ctd,+}(O_i;G_0,S^*) +O(1)\}
\eqref{eq:eqlem4} = \frac{1}{nh} \sum_{i=1}^n \{\mathbf{X}_{i,h}L_{ih}^+\{Y_{DR^*,i}(O_i;\hat{G},\hat{S}) - E(Y_{DR^*,i}(O_i;G_0,S^*)|W_i) +E_n^*\} \\
= \frac{1}{nh} \sum_{i=1}^n \{\mathbf{A}_{i,DR}^{ctd,+}(O_i;G_0,S^*) +\mathbf{X}_{i,h}L_{ih}^+ E_n^*\}
\end{gather*}
%By calculations similar to those in Hahn, Todd, and Van der Klaauw (1999), we can see that
%\begin{gather*}
%\frac{1}{nh} \sum_{i=1}^n \mathbf{\bar{A}}_{DR^*,h}^+(\mathcal{O};G_0) = \frac{1}{nh} \sum_{i=1}^n \mathbf{A}_{i,IPCW}^{ctd,+}(O_i;G_0). 
%(nh)^{-1} \sum_{i=1}^n (\bar{Y}_{DR^*,h}(G_0,S) - \bar{A}_{DR^*,h}(G_0,S)) = \frac{1}{nh} \sum_{i=1}^n \mathbf{A}_{i,DR}^{ctd}(O_i;G,S)
%\end{gather*} 
By looking at variances of the individual terms and and assumption (R9),
\begin{gather*}
Var \bigg{\{}\frac{1}{nh} \sum_{i=1}^n \mathbf{\bar{A}}_{DR^*,h}^+(\mathcal{O};\hat{G},\hat{S}) \bigg{\}} = Var\bigg{\{}\frac{1}{nh} \sum_{i=1}^n \{\mathbf{A}_{i,DR}^{ctd,+}(O_i;G_0,S^*)  +\mathbf{X}_{i,h}L_{ih}^+ E_n^*\} \bigg{\}}\\
= \frac{1}{(nh)^2} \bigg{[}\int_{w_0}^{\infty} \bigg{(}\frac{w-w_0}{h} \bigg{)}^{2q} \bigg{\{}I(W_i \geq w_0)K\bigg{(}\frac{w-w_0}{h}\bigg{)}\bigg{\}}^2 \sigma_{DR}^{2+}(w;G_0,S^*) g(w)dw + O(1) + O(1)\bigg{]}\\
= \frac{1}{nh} \{\sigma_{DR}^{2+}(w_0;G_0,S^*)v_q + o(1) + o(1)\} = \frac{1}{nh} \{\sigma_{DR}^{2+}(w_0;G_0,S^*)v_q + o(1)\}. 
\end{gather*}
Note that the covariance term between $\mathbf{A}_{i,DR}^{ctd,+}(O_i;G_0,S^*)$ and $\mathbf{X}_{i,h}L_{ih}^+ E_n^*$ is finite, hence it is $O(1)$. Similarly, 
\begin{gather*}
Var \bigg{\{}\frac{1}{nh} \sum_{i=1}^n \mathbf{\bar{A}}_{IPCW^*,h}^+(\mathcal{O};\hat{G})\bigg{\}} = \frac{1}{nh} \{\sigma_{IPCW}^{2+}(w_0;G_0)v_q + o(1)\}.
\end{gather*}
Derivation of the variance of $\mathbf{\bar{A}}_{Z^*,h}^+$ is similar. Next, for cross-covariance quantities, we only prove for individual terms of $Cov(\mathbf{\bar{A}}_{DR^*,h}^+(\mathcal{O};\hat{G},\hat{S}), \mathbf{\bar{A}}_{Z^*,h}^+)$. Note that 
\begin{gather*}
\frac{1}{nh^2} \int_{w_0}^{\infty} \bigg{(}\frac{w-w_0}{h} \bigg{)}^{2q} \bigg{\{}I(W_i \geq w_0)K\bigg{(}\frac{w-w_0}{h}\bigg{)}\bigg{\}}^2 \eta_{DR}(w) g(w)dw \\
= \frac{1}{nh}\{\eta_{DR}^+(w_0;G_0,S^*)v_q+o(1)\}.
\end{gather*}
Hence the result follows.
\end{proof}
\begin{lemma}
\begin{gather*}
\sqrt{nh} \begin{pmatrix} \mathbf{\bar{A}}_{IPCW^*,h}^+(\mathcal{O};G_0) \\ \mathbf{\bar{A}}_{Z^*,h}^+ \end{pmatrix}
\overset{d}{\longrightarrow} \{g(w_0)\}\mathcal{N}\bigg{(}0,\begin{pmatrix} \sigma_{IPCW}^{2+}(w_0;G_0)\mathcal{V} & \eta_{IPCW}^{+}(w_0;G_0)\mathcal{V} \\ \eta_{IPCW}^{+}(w_0;G_0)\mathcal{V} & p^+(w_0)(1-p^+(w_0))\mathcal{V} \end{pmatrix} \bigg{)} \\
\sqrt{nh} \begin{pmatrix} \mathbf{\bar{A}}_{DR^*,h}^+(\mathcal{O};G_0,S) \\ \mathbf{\bar{A}}_{Z^*,h}^+ \end{pmatrix}
\overset{d}{\longrightarrow} \{g(w_0)\}\mathcal{N}\bigg{(}0,\begin{pmatrix} \sigma_{DR}^{2+}(w_0;G_0,S)\mathcal{V} & \eta_{DR}^{+}(w_0;G_0,S)\mathcal{V} \\ \eta_{DR}^{+}(w_0;G_0,S)\mathcal{V} & p^+(w_0)(1-p^+(w_0))\mathcal{V} \end{pmatrix} \bigg{)}. 
\end{gather*}
\end{lemma}
\begin{proof}
For any $\mathbf{b} \in \mathbb{R}^4$, we will show that
\begin{gather*}
\mathbf{a}^T \begin{pmatrix} \mathbf{\bar{A}}_{DR^*,h}^+(\mathcal{O};\hat{G},\hat{S}) \\ \mathbf{\bar{A}}_{Z^*,h}^+ \end{pmatrix} \overset{d}{\longrightarrow} \mathbf{a}^T \mathcal{X},
\end{gather*}
where $\mathcal{X}$ has multivariate normal distribution with mean 0 and covariance matrix 
\begin{gather*}
\begin{pmatrix} \sigma_{DR}^{2+}(w_0;G_0,S^*)\mathcal{V} & \eta_{DR}^{+}(w_0;G_0,S^*)\mathcal{V} \\ \eta_{IPCW}^{+}(w_0;G_0)\mathcal{V} & p^+(w_0)(1-p^+(w_0))\mathcal{V}
\end{pmatrix}.
\end{gather*}
As in Hahn, Todd and Van der Klaauw (1999), by using the Lyapounov condition with only considering individual terms, we would like to show that
\begin{gather*}
\frac{1}{\sqrt{nh}}\frac{1}{h} E\bigg{[}(L_{ih}^+)^3\bigg{(}\dfrac{W_i - w_0}{h}\bigg{)}^{3q}|Y_{DR,i}(O_i;\hat{G},\hat{S}) - \mu(W_i)|^3\bigg{]} = o(1)
\end{gather*}
By expansion, 
\begin{gather}
\frac{1}{h}E\bigg{[}(L_{ih}^+)^3\bigg{(}\dfrac{W_i - w_0}{h}\bigg{)}^{3q}|Y_{DR,i}(O_i;\hat{G},\hat{S}) - \mu(W_i)|^3\bigg{]} \nonumber \\
= \frac{1}{h}E\bigg{[}(L_{ih}^+)^3\bigg{(}\dfrac{W_i - w_0}{h}\bigg{)}^{3q}|Y_{DR,i}(O_i;\hat{G},\hat{S}) - Y_{DR,i}(O_i;G_0,S^*) + Y_{DR,i}(O_i;G_0,S^*) - \mu(W_i)|^3\bigg{]} \nonumber \\
= \frac{1}{h}E\bigg{[}(L_{ih}^+)^3\bigg{(}\dfrac{W_i - w_0}{h}\bigg{)}^{3q}\{|Y_{DR,i}(O_i;\hat{G},\hat{S}) - Y_{DR,i}(O_i;G_0,S^*)|^3 \nonumber \\
+ 3\{Y_{DR,i}(O_i;\hat{G},\hat{S}) - Y_{DR,i}(O_i;G_0,S^*)\}^2|Y_{DR,i}(O_i;G_0,S^*) - \mu(W_i)| \nonumber \\
+ 3|Y_{DR,i}(O_i;\hat{G},\hat{S}) - Y_{DR,i}(O_i;G_0,S^*)|\{Y_{DR,i}(O_i;G_0,S^*) - \mu(W_i)\}^2 \nonumber \\
+|Y_{DR,i}(O_i;G_0,S^*) - \mu(W_i)|^3\}\bigg{]}
\label{eq:eqlem5}
\end{gather}
By consistency of $\hat{G}$ and $\hat{S}$, and by assumptions (R6) and (R9),
\begin{gather*}
\eqref{eq:eqlem5}= \int_{0}^{\infty} u^{3q} \{K(u)\}^3 \iota(w_0 + uh) g(w_0+uh) du + O(1).
\end{gather*}
where $\iota(w) = E\{|Y_{DR,i}(O_i;G_0,S^*) - \mu(W_i)|^3|W_i = w\}$. Hence 
\begin{gather*}
\frac{1}{\sqrt{nh}}\frac{1}{h}E\bigg{[}(L_{ih}^+)^3\bigg{(}\dfrac{W_i - w_0}{h}\bigg{)}^{3q}|Y_{IPCW,i}(O_i;G_0) - \mu(W_i)|^3\bigg{]} = o(1).
\end{gather*}
Similarly, $\dfrac{1}{\sqrt{nh}}\dfrac{1}{h}E\bigg{[}(L_{ih}^+)^3\bigg{(}\dfrac{W_i - w_0}{h}\bigg{)}^{3q}|Z_{i} - p(W_i)|^3\bigg{]} = o(1)$. Hence the conditional asymptotic normality of $\begin{pmatrix} \mathbf{\bar{A}}_{DR^*,h}^+(\mathcal{O};G_0,S^*) \\ \mathbf{\bar{A}}_{Z^*,h}^+ \end{pmatrix}$ follows by the Cram\'er-Wold theorem. Similarly, 
\begin{gather*}
\frac{1}{\sqrt{nh}}\frac{1}{h}E\bigg{[}(L_{ih}^+)^3\bigg{(}\frac{W_i - w_0}{h}\bigg{)}^{3q}|Y_{IPCW,i}(\mathcal{O};\hat{G}) - \mu(W_i)|^3\bigg{]} = o(1). 
\end{gather*}
Hence the result for $\begin{pmatrix} \mathbf{\bar{A}}_{IPCW^*,h}^+(\mathcal{O};\hat{G}) \\ \mathbf{\bar{A}}_{Z^*,h}^+ \end{pmatrix}$ follows. 
\end{proof} 
\begin{lemma}
\begin{gather*}
\frac{1}{\sqrt{nh}}\sum_{i=1}^n \begin{pmatrix} \mathbf{A}_{IPCW^*,i,h}^+(O_i;\hat{G}) \\ \mathbf{A}_{Z^*,i,h}^+ \end{pmatrix} - \frac{1}{2}n^{1/2}h^{5/2}g(w_0)\begin{pmatrix} \mu''^+(w_0)\boldsymbol{\delta} \\ p''^+(w_0)\boldsymbol{\delta} \end{pmatrix} \\
\overset{d}{\longrightarrow} \{g(w_0)\}^{1/2}\mathcal{N}\bigg{(}0, \begin{pmatrix} \sigma_{IPCW}^{2+}(w_0;G_0)\mathcal{V} & \eta_{IPCW}^{+}(w_0;G_0)\mathcal{V} \\ \eta_{IPCW}^{+}(w_0;G_0)\mathcal{V} & p^+(w_0)(1-p^+(w_0))\mathcal{V} \end{pmatrix} \bigg{)} \\
\frac{1}{\sqrt{nh}}\sum_{i=1}^n \begin{pmatrix} \mathbf{A}_{DR^*,i,h}^+(O_i;\hat{G},\hat{S}) \\ \mathbf{A}_{Z^*,i,h}^+ \end{pmatrix} - \frac{1}{2}n^{1/2}h^{5/2}g(w_0)\begin{pmatrix} \mu''^+(w_0)\boldsymbol{\delta} \\ p''^+(w_0)\boldsymbol{\delta} \end{pmatrix} \\
\overset{d}{\longrightarrow} \{g(w_0)\}^{1/2}\mathcal{N}\bigg{(}0, \begin{pmatrix} \sigma_{DR}^{2+}(w_0;G_0,S^*)\mathcal{V} & \eta_{DR}^{+}(w_0;G_0,S^*)\mathcal{V} \\ \eta_{DR}^{+}(w_0;G_0,S)\mathcal{V} & p^+(w_0)(1-p^+(w_0))\mathcal{V} \end{pmatrix} \bigg{)}. 
\end{gather*}
\end{lemma}
\begin{proof}
Note that 
\begin{gather*}
\frac{1}{\sqrt{nh}}\sum_{i=1}^n \begin{pmatrix} \mathbf{A}_{DR^*,i,h}^+(O_i;\hat{G},\hat{S}) \\ \mathbf{A}_{Z^*,i,h}^+ \end{pmatrix} = \sqrt{nh} \begin{pmatrix} \mathbf{\bar{A}}_{DR^*,h}^+(\mathcal{O};\hat{G},\hat{S}) \\ \mathbf{\bar{A}}_{Z^*,h}^+ \end{pmatrix} + \frac{1}{\sqrt{nh}} \sum_{i=1}^n \begin{pmatrix} \mathbf{A}_{DR^*,h}^+(W_i;\hat{G},\hat{S}) \\ \mathbf{A}_{Z^*,h}^+(W_i) \end{pmatrix}.
\end{gather*}
By Lemma 1 and 2,
\begin{gather*} 
\frac{1}{\sqrt{nh}}
\sum_{i=1}^n
\begin{pmatrix}
\mathbf{A}_{DR^*,h}^+(W_i;\hat{G},\hat{S}) \\
\mathbf{A}_{Z^*,h}^+(W_i) \\ 
\end{pmatrix}
=
\begin{pmatrix}
\frac{1}{2}g(w_0)(nh)^{1/2}\mu''^{+}(w_0)h^2(\boldsymbol{\delta} + o(1)) + (nh)^{1/2}o_p(h^2) \\
\frac{1}{2}g(w_0)(nh)^{1/2}p''^{+}(w_0)h^2(\boldsymbol{\delta} + o(1)) + (nh)^{1/2}o_p(h^2) 
\end{pmatrix}.
\end{gather*}
Since $h \sim n^{-1/5}$, we obtain
\begin{gather*} 
\frac{1}{\sqrt{nh}} \sum_{i=1}^n
\begin{pmatrix}
\mathbf{A}_{DR^*,h}^+(W_i;\hat{G},\hat{S}) \\
\mathbf{A}_{Z^*,h}^+(W_i)
\end{pmatrix}
\overset{p}{\longrightarrow} 
\begin{pmatrix}
\frac{1}{2}g(w_0)(nh)^{1/2}\mu''^{+}(w_0)h^2\boldsymbol{\delta}  \\
\frac{1}{2}g(w_0)(nh)^{1/2}p''^{+}(w_0)h^2\boldsymbol{\delta}  
\end{pmatrix},
\end{gather*}
and by Lemma 4,
\begin{gather*}
\sqrt{nh} \begin{pmatrix} \mathbf{\bar{A}}_{DR^*,h}^+(\mathcal{O};\hat{G},\hat{S}) \\ \mathbf{\bar{A}}_{Z^*,h}^+ \end{pmatrix} 
\overset{d}{\longrightarrow} \{g(w_0)\}\mathcal{N}\bigg{(}0,\begin{pmatrix} \sigma_{DR}^{2+}(w_0;G_0,S^*)\mathcal{V} & \eta_{DR}^{+}(w_0;G_0,S^*)\mathcal{V} \\ \eta_{DR}^{+}(w_0;G_0,S^*)\mathcal{V} & p^+(w_0)(1-p^+(w_0))\mathcal{V} \end{pmatrix} \bigg{)}. 
\end{gather*}
Then by Slutsky's theorem, the result holds for the doubly robust estimator. Similarly, we can prove for $\dfrac{1}{nh} \displaystyle \sum_{i=1}^n \mathbf{A}_{IPCW^*,i,h}^+(O_i;\hat{G})$.
\end{proof}
\noindent For all lemmas, we can derive similar results for $W < w_0$. By combining the results from the cases $W \geq w_0$ and $W < w_0$, we are ready to prove the main theorems. Recall that 
\begin{gather*}
\tau^Y = \lim_{w \downarrow w_0} E(Y|W = w) - \lim_{w \uparrow w_0} E(Y|W = w) \\
 \tau^Z = \lim_{w \downarrow w_0} P(Z=1|W = w) - \lim_{w \uparrow w_0} P(Z=1|W = w) \\
\hat{\tau}_{FRD}^{IPCW}(G) = \frac{\hat{\alpha}_{R,IPCW}^{FRD,Y}(G) - \hat{\alpha}_{L,IPCW}^{FRD,Y}(G)}{\hat{\alpha}_R^{FRD,Z} - \hat{\alpha}_L^{FRD,Z}} \quad \hat{\tau}_{FRD}^{DR}(G,S) = \frac{\hat{\alpha}_{R,DR}^{FRD,Y}(G,S) - \hat{\alpha}_{L,DR}^{FRD,Y}(G,S)}{\hat{\alpha}_R^{FRD,Z} - \hat{\alpha}_L^{FRD,Z}} \\
\hat{\tau}_{SRD}^{IPCW}(G) = \hat{\alpha}_{R,IPCW}^{SRD,Y}(G) - \hat{\alpha}_{L,IPCW}^{SRD,Y}(G) \quad \hat{\tau}_{SRD}^{DR}(G,S) = \hat{\alpha}_{R,DR}^{SRD,Y}(G,S) - \hat{\alpha}_{L,DR}^{SRD,Y}(G,S).
\end{gather*}
Let 
\begin{gather*}
\rho^+ = \frac{\displaystyle \bigg{(}\int_0^{\infty} u^2K(u)du\bigg{)}^2 - \bigg{(}\int_0^{\infty} u^3K(u)du\bigg{)}\bigg{(}\int_0^{\infty} uK(u)du\bigg{)}}{2\displaystyle\bigg{(}\int_0^{\infty} u^2K(u)du\bigg{(}\int_0^{\infty} K(u)du\bigg{)} - \bigg{(}\int_0^{\infty} uK(u)du\bigg{)}^2\bigg{)}} \\
\rho^- = \frac{\displaystyle \bigg{(}\int_{-\infty}^{0} u^2K(u)du\bigg{)}^2 - \bigg{(}\int_{-\infty}^{0} u^3K(u)du\bigg{)}\bigg{(}\int_{-\infty}^{0} uK(u)du\bigg{)}}{2\displaystyle\bigg{(}\int_{-\infty}^{0} u^2K(u)du\bigg{(}\int_{-\infty}^{0} K(u)du\bigg{)} - \bigg{(}\int_{-\infty}^{0} uK(u)du\bigg{)}^2\bigg{)}} \\
\upsilon^+ = \frac{\displaystyle \int_0^{\infty}\bigg{\{}\bigg{(}\int_0^{\infty} s^2K(s)ds\bigg{)} - \bigg{(}\int_0^{\infty} sK(s)ds\bigg{)}u\bigg{\}}^2 \{K(u)\}^2 du}{\displaystyle g(w_0)\bigg{\{}\bigg{(}\int_0^{\infty} u^2K(u)du\bigg{)}\bigg{(}\int_0^{\infty} K(u)du\bigg{)} - \bigg{(}\int_0^{\infty} uK(u)du\bigg{)}^2\bigg{\}}^2} \\
\upsilon^- = \frac{\displaystyle \int_{-\infty}^{0}\bigg{\{}\bigg{(}\int_{-\infty}^{0} s^2K(s)ds\bigg{)} - \bigg{(}\int_{-\infty}^{0} sK(s)ds\bigg{)}u\bigg{\}}^2 \{K(u)\}^2 du}{\displaystyle g(w_0)\bigg{\{}\bigg{(}\int_{-\infty}^{0} u^2K(u)du\bigg{)}\bigg{(}\int_{-\infty}^{0} K(u)du\bigg{)} - \bigg{(}\int_{-\infty}^{0} uK(u)du\bigg{)}^2\bigg{\}}^2}.
\end{gather*}
Moreover, let
\begin{gather*}
\varphi_{FRD} = \frac{1}{\tau^Z}(\rho^+\mu''^+(w_0) - \rho^-\mu''^-(w_0)) - \frac{\tau^{Y}}{(\tau^Z)^2}(\rho^+p''^+(w_0) - \rho^-p''^-(w_0)) \\
\Sigma_{FRD}^{IPCW}(G) = \frac{1}{\tau^Z}(\upsilon^+\sigma_{IPCW}^{2+}(w_0;G)+\upsilon^-\sigma_{IPCW}^{2-}(w_0;G)) -2 \frac{\tau^Y}{(\tau^Z)^3}(\upsilon^+\eta_{IPCW}^{+}(w_0;G)+\\ \upsilon^-\eta_{IPCW}^{-}(w_0;G))  + \frac{(\tau^Y)^2}{(\tau^Z)^4}\{\upsilon^+p^+(w_0)(1-p^+(w_0)) + \upsilon^-p^-(w_0)(1-p^-(w_0))\} \\
\Sigma_{FRD}^{DR}(G,S) = \frac{1}{\tau^Z}(\upsilon^+\sigma_{DR}^{2+}(w_0;G,S)+\upsilon^-\sigma_{DR}^{2-}(w_0;G,S)) -2 \frac{\tau^Y}{(\tau^Z)^3}(\upsilon^+\eta_{DR}^{+}(w_0;G,S)+\\ \upsilon^-\eta_{DR}^{-}(w_0;G,S)) + \frac{(\tau^Y)^2}{(\tau^Z)^4}\{\upsilon^+p^+(w_0)(1-p^+(w_0)) + \upsilon^-p^-(w_0)(1-p^-(w_0))\}. 
\end{gather*}
and 
\begin{gather*}
\varphi_{SRD} = \rho^+\mu''^+(w_0) - \rho^-\mu''^-(w_0)  \\
\Sigma_{SRD}^{IPCW}(G) = \upsilon^+\sigma_{IPCW}^{2+}(w_0;G)+\upsilon^-\sigma_{IPCW}^{2-}(w_0;G)  \\
\Sigma_{SRD}^{DR}(G,S) = \upsilon^+\sigma_{DR}^{2+}(w_0;G,S)+\upsilon^-\sigma_{DR}^{2-}(w_0;G,S). 
\end{gather*}
%\begin{customthm}{1}
%Assume that conditions (C1)-(C5), (R1)-(R8) hold. By Lemma 1-6 in Appendix,  
%\begin{gather*}
%n^{2/5}(\hat{\tau}_{FRD}^{IPCW}(G_0) - \tau_{FRD}- \varphi_{FRD}) \overset{d}{\longrightarrow} N(0,\Sigma_{FRD}^{IPCW}(G_0)) \\
%n^{2/5}(\hat{\tau}_{FRD}^{DR}(G_0,S) - \tau_{FRD}- \varphi_{FRD}) \overset{d}{\longrightarrow} N(0,\Sigma_{FRD}^{DR}(G_0,S)) \\
%n^{2/5}(\hat{\tau}_{FRD}^{DR}(G,S_0) - \tau_{FRD}- \varphi_{FRD}) \overset{d}{\longrightarrow} N(0,\Sigma_{FRD}^{DR}(G,S_0)) \\
%n^{2/5}(\hat{\tau}_{FRD}^{DR}(G_0,S_0) - \tau_{FRD}- \varphi_{FRD}) \overset{d}{\longrightarrow} N(0,\Sigma_{FRD}^{DR}(G_0,S_0)) 
%\end{gather*}
%\end{customthm}
\begin{proof}[Proof of Theorem 1]
Now we use a matrix to express terms. As before, we only prove for doubly robust estimator. Let $\mathbf{Y}_{DR}(\mathcal{O};\hat{G},\hat{S}) = \{Y_{DR,i}(O_i;\hat{G},\hat{S})\}_{i=1}^n$ and $\mathbf{Z} = \{Z_i\}_{i=1}^n$. Note that 
\begin{gather*}
\sum_{i=1}^n \begin{pmatrix} \mathbf{A}_{DR^*,i,h}^+(O_i;\hat{G},\hat{S}) \\ \mathbf{A}_{Z^*,i,h}^+ \end{pmatrix} = \begin{pmatrix} \mathbf{X}_h^T\mathbf{W}_h^+\mathbf{Y}_{DR}(\mathcal{O};\hat{G},\hat{S}) \\ \mathbf{X}_h^T\mathbf{W}_h^+\mathbf{Z} \end{pmatrix}
\end{gather*}
Similar calculation as Hahn, Todd and Van der Klaauw (1999), 
\begin{gather*}
n^{2/5}
\begin{pmatrix}
\hat{\alpha}_{R,DR}^{FRD,Y}(\hat{G},\hat{S}) - \mu^+(w_0) \\
\hat{\beta}_{R,DR}^{FRD,Y}(\hat{G},\hat{S}) - \mu'^+(w_0) \\
\hat{\alpha}_R^{FRD,Z} - p^+(w_0) \\
\hat{\beta}_R^{FRD,Z} - p'^+(w_0)
\end{pmatrix}
= \boldsymbol{\Psi}
\begin{pmatrix}
(\mathbf{X}_h^T\mathbf{W}_h^+\mathbf{X}_h)^{-1}\mathbf{X}_h^T\mathbf{W}_h^+\mathbf{Y}_{DR}(\mathcal{O};\hat{G},\hat{S}) \\
(\mathbf{X}_h^T\mathbf{W}_h^+\mathbf{X}_h)^{-1}\mathbf{X}_h^T\mathbf{W}_h^+\mathbf{Z}
\end{pmatrix},
\end{gather*}
where $\boldsymbol{\Psi} = \begin{pmatrix} 1 & 0 & 0 & 0\\ 0 & h^{-1} & 0 & 0 \\ 0 & 0 & 1 & 0\\ 0 & 0 & 0 & h^{-1} \end{pmatrix}$.
Then by Lemma 6 with Lemma 1,
\begin{gather*}
n^{2/5}
\begin{pmatrix}
\hat{\alpha}_{R,DR}^{FRD,Y}(\hat{G},\hat{S}) - \mu^+(w_0) \\
\hat{\beta}_{R,DR}^{FRD,Y}(\hat{G},\hat{S}) - \mu'^+(w_0) \\
\hat{\alpha}_R^{FRD,Z} - p^+(w_0) \\
\hat{\beta}_R^{FRD,Z} - p'^+(w_0)
\end{pmatrix}
- \frac{1}{2}\boldsymbol{\Psi}^{-1}\begin{pmatrix} \boldsymbol{\Omega}^{-1} & 0 \\ 0 & \boldsymbol{\Omega}^{-1}\end{pmatrix}\begin{pmatrix}  \mu''^+(w_0) \boldsymbol{\delta} \\ p''^+(w_0) \boldsymbol{\delta}\end{pmatrix} \overset{d}{\longrightarrow} \\
\mathcal{N}\bigg{(}0, g(w_0)^{-1}\boldsymbol{\Psi}^{-1}\begin{pmatrix} \boldsymbol{\Omega}^{-1} & 0 \\ 0 & \boldsymbol{\Omega}^{-1}\end{pmatrix}\begin{pmatrix}\sigma_{DR}^{2+}(w_0;G_0,S^*)\mathcal{V} & \eta_{DR}^{+}(w_0;G_0,S^*)\mathcal{V} \\ \eta_{DR}^{+}(w_0;G_0,S^*)\mathcal{V} & p^{+}(w_0)(1-p^{+}(w_0))\mathcal{V} \end{pmatrix} \begin{pmatrix} \boldsymbol{\Omega}^{-1} & 0 \\ 0 & \boldsymbol{\Omega}^{-1}\end{pmatrix} \boldsymbol{\Psi}^{-1}  \bigg{)}. 
\end{gather*}
Then 
\begin{gather*}
\hspace*{-20mm}
n^{2/5}
\bigg{(}
\begin{pmatrix}
\hat{\alpha}_{R,DR}^{FRD,Y}(\hat{G},\hat{S}) - \mu^+(w_0) \\
\hat{\alpha}_R^{FRD,Z} - p^+(w_0) 
\end{pmatrix}
- \rho^+ \begin{pmatrix} \mu''^+(w_0) \\ p''^+(w_0) \end{pmatrix}\bigg{)} \overset{d}{\longrightarrow} \mathcal{N}\bigg{(}0, \upsilon^+\begin{pmatrix} \sigma_{DR}^{2+}(w_0;G_0,S^*) & \eta_{DR}^+(w_0;G_0,S^*) \\ \eta_{DR}^+(w_0;G_0,S^*) & p^+(w_0)(1-p^+(w_0)) \end{pmatrix}\bigg{)}.
\end{gather*}
Similarly,
\begin{gather*}
\hspace*{-22.5mm}
n^{2/5}
\bigg{(}
\begin{pmatrix}
\hat{\alpha}_{L,DR}^{FRD,Y}(\hat{G},\hat{S}) - \mu^-(w_0) \\
\hat{\alpha}_L^{FRD,Z} - p^-(w_0) \\
\end{pmatrix}
- \rho^- \begin{pmatrix} \mu''^-(w_0) \\ p''^-(w_0) \end{pmatrix}\bigg{)} \overset{d}{\longrightarrow} \mathcal{N}\bigg{(}0, \upsilon^-\begin{pmatrix} \sigma_{DR}^{2-}(w_0;G_0,S^*) & \eta_{DR}^-(w_0;G_0,S^*) \\ \eta_{DR}^-(w_0;G_0,S^*) & p^-(w_0)(1-p^-(w_0)) \end{pmatrix}\bigg{)}.
\end{gather*}
Then
\begin{gather*}
\hspace*{-20mm}
n^{2/5}\bigg{(}
\begin{pmatrix}
\hat{\alpha}_{R,DR}^{FRD,Y}(\hat{G},\hat{S}) - \hat{\alpha}_{L,DR}^{FRD,Y}(\hat{G},\hat{S}) - (\mu^+(w_0) - \mu^-(w_0))\\
\hat{\alpha}_R^{FRD,Z} - \hat{\alpha}_L^{FRD,Z} - (p^+(w_0) - p^-(w_0))\\
\end{pmatrix}
\bigg{)}
-  \begin{pmatrix} \rho^+ \mu''^+(w_0) - \rho^- \mu''^-(w_0) \\ \rho^+ p''^+(w_0) - \rho^- p''^-(w_0) \end{pmatrix}  
\overset{d}{\longrightarrow} \mathcal{N}(0, \Xi),
\end{gather*}
where
\begin{gather*}
\Xi = \upsilon^+\begin{pmatrix} \sigma_{DR}^{2+}(w_0;G_0,S^*) & \eta_{DR}^+(w_0;S^*) \\ \eta_{DR}^+(w_0;G_0,S^*) & p^+(w_0)(1-p^+(w_0)) \end{pmatrix} + \upsilon^-\begin{pmatrix} \sigma_{DR}^{2-}(w_0;G_0,S^*) & \eta_{DR}^-(w_0;G_0,S^*) \\ \eta_{DR}^-(w_0;G_0,S^*) & p^-(w_0)(1-p^-(w_0)) \end{pmatrix}.
\end{gather*}
By the delta method, the result follows. The proof for $\hat{\tau}_{FRD}^{IPCW}(\hat{G})$ is similar. 
\end{proof}
\noindent From Theorem 1, we can derive Corollary 1 because the corollary is a special case of Theorem 1. Next, we prove Theorem 2.

\begin{proof}[Proof of Theorem 2]
By Steingrimsson, Diao, and Strawderman (2019) and Suzukawa (2004), variance formulae for IPCW and DR estimator are given by 
\begin{gather*}
\sigma_{IPCW}^2(w;G_0)  = Var(Y_{IPCW}(O;G_0)|W=w) = \int_0^{\infty} \frac{Y^2}{G_0(u)}dF_0(u|w) - \{\mu(w)\}^2\\
\sigma_{DR}^2(w;G_0,S) = Var(Y_{DR}(O;G_0,S)|W=w) = \int_0^{\infty} \frac{Y^2}{G_0(u)}dF_0(u|w)  - \\ \int_0^{\infty} \frac{S_0(u|w)\{Q_Y(u,w,S)(Q_Y(u,w,S) - 2Q_Y(u,w,S_0))\}}{\{G_0(u)\}^2}d\bar{G}_0(u) - \{\mu(w)\}^2\\
%\sigma_{DR}^2(w,G,S_0) = Var(Y_{DR}(G,S_0)|W=w) = \int_0^{\infty} \frac{Y^2}{G(u|w)}dF_0(u|w)  - \\ \int_0^{\infty} \frac{S_0(x|w)\{m^2(x,w,S_0)\}}{G^2(x|w)}d\bar{G}_0(x|w) - \mu(w)^2\\
\sigma_{DR}^2(w;G_0,S_0) = Var(Y_{DR}(O;G_0,S_0)|W=w) = \int_0^{\infty} \frac{Y^2}{G_0(u)}dF_0(u|w)  - \\ \int_0^{\infty} \frac{S_0(u|w)\{Q(u,w,S_0)\}^2}{\{G_0(u)\}^2}d\bar{G}_0(u) - \{\mu(w)\}^2.
\end{gather*}
Hence for any $w$, we have $\sigma_{DR}^2(w,G_0,S_0) \leq \min\{\sigma_{IPCW}^2(w;G_0), \sigma_{DR}^2(w;G_0,S)\}$. Clearly,
\begin{gather*}
\sigma_{DR}^{2+}(w_0;G_0,S_0) \leq \min\{\sigma_{IPCW}^{2+}(w_0;G_0),\sigma_{DR}^{2+}(w_0;G_0,S)\} \\
\sigma_{DR}^{2-}(w_0;G_0,S_0) \leq \min\{\sigma_{IPCW}^{2-}(w_0;G_0),\sigma_{DR}^{2-}(w_0;G_0,S)\}.
\end{gather*} 
Since the $\upsilon^+$ and $\upsilon^-$ are positive, the result holds. For the fuzzy RD estimator,
\begin{gather*}
\Sigma_{SRD}^{DR}(G_0,S_0) - \Sigma_{SRD}^{IPCW}(G_0) = \frac{1}{\tau^Z}[\upsilon^+\{\sigma_{DR}^{2+}(w_0;G_0,S_0) - \sigma_{IPCW}^{2+}(w_0;G_0)\} \\
+\upsilon^-\{\sigma_{DR}^{2-}(w_0;G_0,S_0) - \sigma_{IPCW}^{2-}(w_0;G_0)\}] -2 \frac{\tau^Y}{(\tau^Z)^3}[\upsilon^+\{\eta_{DR}^{+}(w_0;G_0,S_0) - \eta_{IPCW}^{+}(w_0;G_0)\} \\
+\upsilon^-\{\eta_{DR}^{-}(w_0;G_0,S_0) - \eta_{IPCW}^{-}(w_0;G_0)\}]. 
\end{gather*} 
Since calculation at sharp RD case, $\sigma_{DR}^{2+}(w_0;G_0,S_0) \leq \sigma_{IPCW}^{2+}(w_0;G_0)$ and $\sigma_{DR}^{2-}(w_0;G_0,S_0) \leq \sigma_{IPCW}^{2-}(w_0;G_0)$. Moreover, from conditions $\eta_{DR}^{+}(w_0;G_0,S_0) \leq \eta_{IPCW}^{+}(w_0;G_0)$ and $\eta_{DR}^{-}(w_0;G_0,S_0) \leq \eta_{IPCW}^{-}(w_0;G_0)$, with $\tau^Y > 0$ and $\tau^Z > 0$, we have $\Sigma_{FRD}^{DR}(G_0,S_0) \leq \Sigma_{FRD}^{IPCW}(G_0)$. Similarly, $\Sigma_{FRD}^{DR}(G_0,S_0) \leq \Sigma_{FRD}^{DR}(G_0,S)$. Hence
\begin{gather*} 
AVar(\hat{\tau}_{FRD}^{DR}(G_0,S_0)) \leq \min\{AVar(\hat{\tau}_{FRD}^{IPCW}(G_0)),AVar(\hat{\tau}_{FRD}^{DR}(G_0,S))\}.
\end{gather*}
\end{proof}


\begin{thebibliography}{}

\bibitem{} Abadie, A. and Imbens, G. W. (2006). Large sample properties of matching estimators for average treatment effects. {\em Econometrica} {\bf 74}, 235-267.
\bibitem{} Andriole, G. L., Crawford, E. D., Grubb III, R. L., Buys, S. S., Chia, D., Church, T. R., Fouad, M. N., Gelmann, E. P., Kvale, P. A., Reding, D. J. and Weissfeld, J. L. (2009). Mortality results from a randomized prostate-cancer screening trial. {\em New England Journal of Medicine} {\bf 360}, 1310-1319.
\bibitem{} Angrist, J. D., Imbens, G. W. and Rubin, D. B. (1996). Identification of causal effects using instrumental variables. {\em Journal of the American statistical Association}, 
{\bf 91}, 444-455.
\bibitem{} Bai, X., Tsiatis, A. A. and O'Brien, S. M. (2013). Doubly-robust estimators of treatment-specific survival distributions in observational studies with stratified sampling. {\em Biometrics} {\bf 69}, 830-839.
\bibitem{} Bor, J., Moscoe, E., Mutevedzi, P., Newell, M. L. and B\"{a}rnighausen, T. (2014). Regression discontinuity designs in epidemiology: causal inference without randomized trials. {\em Epidemiology} {\bf 25}, 729-737.
%\bibitem{} Buckley, J. and James, I. (1979). Linear regression with censored data. {\em Biometrika} 66(3), 429-436.
\bibitem{} Calonico, S., Cattaneo, M. D. and Titiunik, R. (2014). Robust nonparametric confidence intervals for regression-€discontinuity designs. {\em Econometrica} {\bf 82}, 2295-2326.
\bibitem{} Calonico, S., Cattaneo, M. D. and Titiunik, R. (2015a). Optimal data-driven regression discontinuity plots. {\em Journal of the American Statistical Association},  {\bf 110}, 1753-1769.
\bibitem{} Calonico, S., Cattaneo, M. D. and Titiunik, R. (2015b). rdrobust: An R package for robust nonparametric inference in regression-discontinuity designs. {\em R Journal} {\bf 7}, 38-51.
\bibitem{} Calonico, S., Cattaneo, M. D., Farrell, M. H. and Titiunik, R. (2018). Regression discontinuity designs using covariates. {\em Review of Economics and Statistics} 1-10.
\bibitem{} Cattaneo, M. D., Titiunik, R., Vazquez-Bare, G. and Keele, L. (2016). Interpreting regression discontinuity designs with multiple cutoffs. {\em The Journal of Politics} 78(4), 1229-1248.
\bibitem{} Fan, J. and Gijbels, I. (1994). Censored regression: local linear approximations and their applications. {\em Journal of the American Statistical Association} {\bf 89}, 560-570.
\bibitem{} Fan, J. and Gijbels, I. (1996). {\em Local polynomial modelling and its applications: monographs on statistics and applied probability 66}. CRC Press.
\bibitem{} Hahn, J., Todd, P. and Van der Klaauw, W. (1999). Evaluating the effect of an antidiscrimination law using a regression-discontinuity design (No. w7131). Technical Report, National bureau of economic research.
\bibitem{} Hahn, J., Todd, P. and Van der Klaauw, W. (2001). Identification and estimation of treatment effects with a regression-discontinuity design. {\em Econometrica} {\bf 69}, 201-209.
%\bibitem{} Heuchenne, C. and Van Keilegom, I. (2007). Polynomial regression with censored data based on preliminary nonparametric estimation. %{\em Annals of the Institute of Statistical Mathematics}
%{\em Ann. Inst. Statist. Math} 59(2), 273-297.
\bibitem{} Imbens, G. and Kalyanaraman, K. (2012). Optimal bandwidth choice for the regression discontinuity estimator. {\em The Review of economic studies}, {\bf 79}, 933-959.
\bibitem{} Imbens, G. and Zajonc, T. (2009). Regression discontinuity design with vector-argument assignment rules. {\em Unpublished paper}.
\bibitem{} Imbens, G. W. and Lemieux, T. (2008). Regression discontinuity designs: A guide to practice. %{\em Journal of econometrics}, 
{\em J. Econometrics} {\bf 142}, 615-635.
\bibitem{} Lee, D. S. and Lemieux, T. (2010). Regression discontinuity designs in economics. {\em Journal of economic literature} {\bf 48}, 281-355.
%\bibitem{} Lai, T. L. and Ying, Z. (1991). Large sample theory of a modified Buckley-James estimator for regression analysis with censored data. %{\em The Annals of Statistics}, 
%{\em Ann. Statist.} 19(3), 1370-1402.
\bibitem{} Ludwig, J. and Miller, D. L. (2005). Does head start improve children's life chances? evidence from a regression discontinuity design (No. w11702). Technical report, National Bureau of Economic Research.
\bibitem{} Ludwig, J. and Miller, D. L. (2007). Does Head Start improve children's life chances? Evidence from a regression discontinuity design. {\em The Quarterly journal of economics}, {\bf 122}, 159-208.
\bibitem{} McCrary, J. (2008). Manipulation of the running variable in the regression discontinuity design: A density test. {\em Journal of econometrics},  {\bf 142}, 698-714.
\bibitem{} Moscoe, E., Bor, J. and B\"{a}rnighausen, T. (2015). Regression discontinuity designs are underutilized in medicine, epidemiology, and public health: a review of current and best practice. {\em Journal of clinical epidemiology}, {\bf 68}, 132-143.
%\bibitem{} Ritov, Y. (1990). Estimation in a linear regression model with censored data. %{\em The Annals of Statistics}, 
%{\em Ann. Statist.} 18(1), 303-328.
\bibitem{} Rubin, D. and Van der Laan, M. J. (2007). A doubly robust censoring unbiased transformation. %{\em The international journal of biostatistics}, 
{\em Int. J. Biostat.} {\bf 3}, 1-19.
\bibitem{} Shoag, J., Halpern, J., Eisner, B., Lee, R., Mittal, S., Barbieri, C. E. and Shoag, D. (2015). Efficacy of prostate-specific antigen screening: Use of regression discontinuity in the PLCO cancer screening trial. {\em JAMA oncology}, {\bf 1}, 984-986.
\bibitem{} Steingrimsson, J. A., Diao, L., Molinaro, A. M. and Strawderman, R. L. (2016). Doubly robust survival trees. {\em Statistics in medicine}, {\bf 35}, 3595-3612.
\bibitem{} Steingrimsson, J. A., Diao, L. and Strawderman, R. L. (2019). Censoring unbiased regression trees and ensembles. {\em Journal of the American Statistical Association}, {\bf 114}, 370-383.
\bibitem{} Suzukawa, A. (2004). Unbiased estimation of functionals under random censorship. {\em Journal of the Japan Statistical Society}, {\bf 34}, 153-172.
\bibitem{} Thistlethwaite, D. L. and Campbell, D. T. (1960). Regression-discontinuity analysis: An alternative to the ex post facto experiment. {\em Journal of Educational Psychology} {\bf 51}, 309-317.
\bibitem{} Tsiatis, A. (2007). {\em Semiparametric theory and missing data.} Springer Science \& Business Media.
%\bibitem{} Wang, N., Carroll, R. J. and Lin, X. (2005). Efficient semiparametric marginal estimation for longitudinal/clustered data. %{\em Journal of the American Statistical Association}, 
%{\em J. Amer. Statist. Assoc} 100(469), 147-157.
\bibitem{} Yang, M. (2013). Treatment effect analyses through orthogonality conditions implied by a fuzzy regression discontinuity design, with two empirical studies. {\em Dep. Econ. Rauch Bus. Cent. Lehigh Univ., Bethlehem, PA} {\bf 308}, 1-61.
\bibitem{} Zajonc, T. (2012). {\em Essays on Causal Inference for Public Policy}. PhD thesis, Harvard University.
%\bibitem[\protect\citeauthoryear{Curtis}{Curtis}{1943}]{1943}
%Curtis, M. (1943).
%\newblock {\em Documents on International Affairs, 1938}, Volume~II.
%\newblock London: Oxford University Press.

%\bibitem[\protect\citeauthoryear{Eubank}{Eubank}{2004}]{Eubank}
%Eubank, K. (2004).
%\newblock {\em The origins of World War II\/} (3 ed.).
%\newblock Wheeling, Ill.: Harlan Davidson.

%\bibitem[\protect\citeauthoryear{Gellately}{Gellately}{1988}]{1988}
%Gellately, R. (1988).
%\newblock The gestapo and german society: Political denunciation in the gestapo
  %case files.
%\newblock {\em The Journal of Modern History\/}~{\em 60}, pp. 654--694.

%\bibitem[\protect\citeauthoryear{Noakes and Pridham}{Noakes and
%  Pridham}{2001}]{Noakes}
%Noakes, J. and G.~Pridham (2001).
%\newblock {\em Nazism, 1919-1945. Vol. 3: Foreign Policy, War and Racial
 % Extermination}.
%\newblock Exeter: University of Exeter Press.

\end{thebibliography}
\end{document}